\documentclass[11pt]{article}

\usepackage{times}
\usepackage{graphicx}
\usepackage{amsmath}
\usepackage{amssymb}

\newtheorem{theorem}{Theorem}[section]
\newtheorem{lemma}[theorem]{Lemma}

\newtheorem{corollary}[theorem]{Corollary}
\newtheorem{definition}[theorem]{Definition}
\newtheorem{observation}[theorem]{Observation}
\newtheorem{remark}[theorem]{Remark}

\newcommand{\sq}{\hbox{\rlap{$\sqcap$}$\sqcup$}}
\newcommand{\qed}{\hspace*{\fill}\sq}
\newenvironment{proof}{\noindent {\bf Proof.}\ }{\qed\par\vskip 4mm\par}
\newenvironment{proofof}[1]{\bigskip \noindent {\bf Proof of #1:}\quad }
{\qed\par\vskip 4mm\par}

\begin{document}

\begin{titlepage}

\title{On Dynamics in Selfish Network Creation\footnote{See \cite{KL13} for the original publication. Permission to make digital or hard copies of all or part of this work for personal or classroom use is granted without fee provided that copies are not made or distributed for profit or commercial advantage and that copies bear this notice and the full citation on the first page. Copyrights for components of this work owned by others than the author(s) must be honored. Abstracting with credit is permitted. To copy otherwise, or republish, to post on servers or to redistribute to lists, requires prior specific permission and/or a fee. Request permissions from permissions@acm.org. Copyright is held by the owner/author(s). Publication rights licensed to ACM.} \\
       {\large(full version)}}

\author{Bernd Kawald \\
   Department of Computer Science\\Humboldt-Universit\"at zu Berlin\\Berlin, Germany \\
   kawald@informatik.hu-berlin.de \\
   \and
   Pascal Lenzner \\
   Department of Computer Science\\Humboldt-Universit\"at zu Berlin\\Berlin, Germany \\
   lenzner@informatik.hu-berlin.de
   }

\date{}

\maketitle \thispagestyle{empty}


\begin{abstract}
\noindent We consider the dynamic behavior of several variants of the Network Creation Game, introduced by Fabrikant et al.~[PODC'03]. Equilibrium networks in these models have desirable properties like low social cost and small diameter, which makes them attractive for the decentralized creation of overlay-networks. Unfortunately, due to the non-constructiveness of the Nash equilibrium, no distributed algorithm for \emph{finding} such networks is known. We treat these games as sequential-move games and analyze whether (uncoordinated) selfish play eventually converges to an equilibrium state. Thus, we shed light on one of the most natural algorithms for this problem: distributed local search, where in each step some agent performs a myopic selfish improving move.

We show that fast convergence is guaranteed for all versions of Swap Games, introduced by Alon et al.~[SPAA'10], if the initial network is a tree. Furthermore, we prove that this process can be sped up to an almost optimal number of moves by employing a very natural move policy. Unfortunately, these positive results are no longer true if the initial network has cycles and we show the surprising result that even one non-tree edge suffices to destroy the convergence guarantee. This answers an open problem from Ehsani et al.~[SPAA'11] in the negative. Moreover, we show that on non-tree networks no move policy can enforce convergence. We extend our negative results to the well-studied original version, where agents are allowed to buy and delete edges as well. For this model we prove that there is no convergence guarantee -- even if all agents play optimally. Even worse, if played on a non-complete host-graph, then there are instances where no sequence of improving moves leads to a stable network. Furthermore, we 
analyze whether cost-sharing has positive impact on the convergence behavior. For this we consider a version by Corbo and Parkes~[PODC'05] where bilateral consent is needed for the creation of an edge and where edge-costs are shared equally among the involved agents. Quite surprisingly, we show that employing such a cost-sharing rule yields even worse dynamic behavior. 

Finally, we contrast our mostly negative theoretical results by a careful empirical study. Our simulations indicate two positive facts: (1) The non-convergent behavior seems to be confined to a small set of pathological instances and is unlikely to show up in practice. (2) In all our simulations we observed a remarkably fast convergence towards a stable network in $\mathcal{O}(n)$ steps, where $n$ is the number of agents.   
\end{abstract}

\bigskip

\centerline{{\bf Keywords}: Network Creation Games, Game Dynamics, Convergence, Stabilization, Distributed Local Search}

\end{titlepage}


\section{Introduction}
Understanding Internet-like networks and their implications on our life is a recent endeavor undertaken by researchers from different research communities. 
Such networks are difficult to analyze since they are created by a multitude of selfish entities (e.g. Internet Service Providers) which modify the infrastructure of parts of the network (e.g. their Autonomous Systems) to improve their service quality. The classical field of Game Theory provides the tools for analyzing such decentralized processes and from this perspective the Internet can be seen as an equilibrium state of an underlying game played by selfish agents. 

Within the last decade several such games have been proposed and analyzed. We will focus on the line of works which consider Network Creation Games, as introduced by Fabrikant et al.~\cite{Fab03}. These games are very simple but they contain an interesting trade-off between an agent's investment in infrastructure and her obtained usage quality. Agents aim to invest as little as possible but at the same time they want to achieve a good connection to all other agents in the network. Network Creation Games and several variants have been studied intensively, but, to the best of our knowledge, almost all these works exclusively focus on properties of the equilibrium states of the game. With this focus, the game is usually considered to be a one-shot simultaneous-move game. 
However, the Internet was not created in ``one shot''. It has rather evolved from an initial network, the ARPANET, into its current shape by repeated infrastructural changes performed by selfish agents who entered or left the stage at some time in the process. For this reason, we focus on a more dynamic point of view: We analyze the properties of the network creation \emph{processes} induced by the sequential-move version of the known models of selfish network creation. 

It is well-known that Network Creation Games have low price of anarchy, which implies that the social cost of the worst stable states arising from selfish behavior is close to the cost of the social optimum. Therefore these games are appealing for the decentralized and selfish creation of networks which optimize the service quality for all agents at low infrastructural cost, e.g. overlay networks created by selfish peers. But, to the best of our knowledge, it is not known how a group of agents can collectively \emph{find} such a desirable stable state. Analyzing the game dynamics of Network Creation Games is equivalent to analyzing a very natural search strategy: (uncoordinated) distributed local search, where in every step some agent myopically modifies the network infrastructure to better suit her needs. Clearly, if at some step in the process no agent wants to modify her part of the network, then a stable network has emerged.    

\subsection{Models and Definitions}
We consider several versions of a network creation process performed by $n$ selfish agents. In all versions we consider networks, where every node corresponds to an agent and undirected links connect network nodes. 
The creation process is based on an underlying Network Creation Game~(NCG) and can be understood as a dynamic process where agents sequentially perform strategy-changes in the NCG. In such games, the strategies of the agents determine which links are present in the network and any strategy-profile, which is a vector of the strategies of all $n$ agents, determines the induced network. But this also works the other way round: Given some network $G = (V,E,o)$, where $V$ is the set of $n$ vertices, $E$ is the set of edges and $o: E \to V$ is the \emph{ownership-function}, which assigns the ownership of an edge to one of its endpoints, then $G$ completely determines the current strategies of all $n$ agents of the NCG. Starting from a network $G_0$, any sequence of strategy-changes by agents can thus be 
seen as a sequence of networks $G_0,G_1,G_2,\dots$, where the network $G_{i+1}$ arises from the network $G_i$ by the strategy-change of exactly one agent. In the following, we will write $xy$ or $yx$ for the undirected edge $\{x,y\} \in E$. In figures we will indicate edge-ownership by directing edges away from their owner. 

The creation process starts in an initial state $G_0$, which we call the \emph{initial network}. A step from state $G_i$ to state $G_{i+1}$ consists of a \emph{move} by one agent. A move of agent $u$ in state $G_i$ is the replacement of agent $u$'s pure strategy in $G_i$ by another \emph{admissible} pure strategy of agent $u$. The induced network after this strategy-change by agent $u$ then corresponds to the state $G_{i+1}$. We consider only improving moves, that is, strategy-changes which strictly decrease the moving agent's cost. The cost of an agent in $G_i$ depends on the structure of $G_i$ and it will be defined formally below. If agent $u$ in state $G_i$ has an admissible new strategy which yields a strict cost decrease for her, then we call agent $u$ \emph{unhappy in network $G_i$} and we let $U_i$ denote the set of all unhappy agents in state $G_i$. Only one agent can actually move in a state of the process and this agent $u \in U_i$, whose move transforms $G_i$ into $G_{i+1}$, is called \emph{the 
moving agent in network $G_i$}. In any state of the process the \emph{move policy} determines which agent is the moving agent. The process stops in some state $G_j$ if no agent wants to perform a move, that is, if $U_j = \emptyset$, and we call the resulting networks \emph{stable}. Clearly, stable networks correspond to pure Nash equilibria of the underlying NCG. 

Depending on what strategies are admissible for an agent in the current state, there are several variants of this process, which we call \emph{game types}:
\begin{itemize}
 \item[$\bullet$] In the \emph{Swap Game}~(SG), introduced as ``Basic Network Creation Game'' by Alon et al.~\cite{ADHL10}, the strategy $S_u$ of an agent $u$ in the network $G_i$ is the set of neighbors of vertex $u$ in $G_i$. The new strategy $S_u^*$ is admissible for agent $u$ in state $G_i$, if $|S_u| = |S_u^*|$ and $|S_u \cap S_u^*| = |S_u|-1$. Intuitively, admissible strategies in the SG are strategies which replace one neighbor $x$ of $u$ by another vertex $y$. Note, that this corresponds to ``swapping'' the edge $ux$ from $x$ towards $y$, which is the replacement of edge $ux$ by edge $uy$. Furthermore, observe, that in any state both endpoints of an edge are allowed to swap this edge. Technically, this means that the ownership of an edge has no influence on the agents' strategies or costs.
 \item[$\bullet$] The \emph{Asymmetric Swap Game}~(ASG), recently introduced by Mihal{\'a}k and Schlegel~\cite{MS12}, is similar to the SG, but here the ownership of an edge plays a crucial role. Only the owner of an edge is allowed to swap the edge in any state of the process. The strategy $S_u$ of agent $u$ in state $G_i$ is the set of neighbors in $G_i$ to which $u$ owns an edge and the strategy $S_u^*$ is admissible for agent $u$ in state $G_i$, if $|S_u| = |S_u^*|$ and $|S_u \cap S_u^*| = |S_u|-1$. Hence, in the ASG the moving agents are allowed to swap one own edge.
 \item[$\bullet$] In the \emph{Greedy Buy Game}~(GBG), recently introduced by us~\cite{L12}, agents have more freedom to act. In any state, an agent is allowed to buy or to delete or to swap one own edge. Hence, the GBG can be seen as an extension of the ASG. The strategy $S_u$ of agent $u$ in state $G_i$ is defined as in the ASG, but the set of admissible strategies is larger: $S_u^*$ is admissible for agent $u$ in state $G_i$ if (1) $|S_u^*| = |S_u|+1$ and $S_u \subset S_u^*$ or (2) if $|S_u^*| = |S_u|-1$ and $S_u^* \subset S_u$ or (3) if $|S_u| = |S_u^*|$ and $|S_u \cap S_u^*| = |S_u|-1$.   
 \item[$\bullet$] The \emph{Buy Game}~(BG), which is the original version of an NCG and which was introduced by Fabrikant et al.~\cite{Fab03}, is the most general version. Here agents can perform arbitrary strategy-changes, that is, agents are allowed to perform any combination of buying, deleting and swapping of own edges. The strategy $S_u$ of agent $u$ in $G_i$ is defined as in the ASG, but an admissible strategy for agent $u$ is any set $S_u^* \subseteq V\setminus \{u\}$.      
\end{itemize}
The \emph{cost} of an agent $u$ in network $G_i$ has the form $c_{G_i}(u) = e_{G_i}(u) + \delta_{G_i}(u)$, where $e_{G_i}(u)$ denotes the \emph{edge-cost} and $\delta_{G_i}(u)$ denotes the \emph{distance-cost} of agent $u$ in the network $G_i$. 
Each edge has cost $\alpha >0$, which is a fixed constant, and this cost has to be paid fully by the owner, if not stated otherwise. Hence, if agent $u$ owns $k$ edges in the network $G_i$, then $e_{G_i}(u) = \alpha k$. 
In the (A)SG we simply omit the edge-cost term in the cost function. 

There are two variants of distance-cost functions which capture the focus on average or worst-case connection quality. In the \textsc{Sum}-version, we have $\delta_{G_i}(u) = \sum_{v \in V(G_i)}d_{G_i}(u,v)$, if the network $G_i$ is connected and $\delta_{G_i}(u) = \infty$, otherwise. In the \textsc{Max}-version, we have $\delta_{G_i}(u) = \max_{v \in V(G_i)}d_{G_i}(u,v)$, if $G_i$ is connected and $\delta_{G_i}(u) = \infty$, otherwise. In both cases $d_{G_i}(u,v)$ denotes the shortest path distance between vertex $u$ and $v$ in the undirected graph $G_i$.

The move policy specifies for any state of the process, which of the unhappy agents is allowed to perform a move. From a mechanism design perspective, the move policy is a way to enforce coordination and to guide the process towards a stable state. We will focus on the \emph{max cost policy}, where the agent having the highest cost is allowed to move and ties among such agents are broken arbitrarily. 
Sometimes we will assume that an adversary chooses the worst possible moving agent. Note, that the move policy only specifies who is allowed to move, not which specific move has to be performed. We do not consider such strong policies since we do not want to restrict the agents' freedom to act. 

Any combination of the four game types, the two distance functions and some move policy together with an initial network completely specifies a network creation process. We will abbreviate names, e.g. by calling the Buy Game with the \textsc{Sum}-version of the distance-cost the \textsc{Sum}-BG. If not stated otherwise, edge-costs cannot be shared. 

A cyclic sequence of networks $C_1,\dots,C_j$, where network $C_{i+1 \bmod j}$ arises from network $C_{i \bmod j}$ by an improving move of one agent is called a \emph{better response cycle}. If every move in such a cycle is a \emph{best response move}, which is a strategy-change towards an admissible strategy which yields the largest cost decrease for the moving agent, then we call such a cycle a \emph{best response cycle}. Clearly, a best response cycle is a better response cycle, but the existence of a better response cycle does not imply the existence of a best response cycle.  

\subsection{Classifying Games According to their Dynamics}
Analyzing the convergence processes of games is a very rich and diverse research area. We will briefly introduce two well-known classes of finite strategic games: \emph{games having the finite improvement property}~(FIPG)~\cite{MS96} and \emph{weakly acyclic games}~(WAG)~\cite{Y93}.

FIPG have the most desirable form of dynamic behavior: Starting from any initial state, every sequence of improving moves must eventually converge to an equilibrium state of the game, that is, such a sequence must have finite length. Thus, in such games distributed local search is guaranteed to succeed. It was shown by Monderer and Shapley~\cite{MS96} that a finite game is a FIPG if and only if there exists a \emph{generalized ordinal potential function} $\Phi$, which maps strategy-profiles to real numbers and has the property that if the moving agent's cost decreases, then the potential function value decreases as well. Stated in our terminology, this means that $\Phi: \mathcal{G}_n \to \mathbb{R}$, where $\mathcal{G}_n$ is the set of all networks on $n$ nodes, and we have $$c_{G_i}(u) - c_{G_{i+1}}(u) >0 \Rightarrow \Phi(G_i) - \Phi(G_{i+1}) >0,$$ if agent $u$ is the moving agent in the network $G_i$. 
Clearly, no FIPG can admit a better response cycle.
An especially nice subclass of FIPG are games that are guaranteed to converge to a stable state in a number of steps which is polynomial in the size of the game. We call this subclass poly-FIPG.

Weakly acyclic games are a super-class of FIPG. Here it is not necessarily true that \emph{any} sequence of improving moves must converge to an equilibrium but we have that from any initial state there exists \emph{some} sequence of improving moves which enforces convergence. Thus, with some additional coordination distributed local search may indeed lead to stable states for such games. A subclass of WAG are games where from any initial state there exists a sequence of best response moves, which leads to an equilibrium. We call those games \emph{weakly acyclic under best response}, BR-WAG for short. Observe, that if a game is not weakly acyclic, then there is \emph{no} way of enforcing convergence if agents stick to playing improving moves.

The above mentioned classes of finite strategic games are related as follows: $$\text{poly-FIPG} \subset \text{FIPG} \subset \text{BR-WAG} \subset \text{WAG}.$$ 
The story does not end here. Very recently, Apt and Simon~\cite{AS12} have classified WAG in much more detail by introducing a ``scheduler'', which is a moderating super-player who guides the agents towards an equilibrium. 

\subsection{Related Work}
The original model of Network Creation Games, which we call the \textsc{Sum}-BG, was introduced a decade ago by Fabrikant et al.~\cite{Fab03}. Their motivation was to understand the creation of Internet-like networks by selfish agents without central coordination. In the following years, several variants were proposed: The \textsc{Max}-BG~\cite{De07}, the \textsc{Sum}-SG and the \textsc{Max}-SG~\cite{ADHL10}, the \textsc{Sum}-ASG and the \textsc{Max}-ASG~\cite{MS12}, the \textsc{Sum}-GBG and the \textsc{Max}-GBG~\cite{L12}, a bounded budget version~\cite{Ehs11}, an edge-restricted version~\cite{De09,Bilo12}, a version with bilateral equal-split cost-sharing~\cite{CP05} and a version considering points in a metric space using a different distance measure~\cite{MSW06}. All these works focus on properties of stable networks or on the complexity of computing an agent's best response. To the best of our knowledge, the dynamic behavior of most of these variants, including best response dynamics in the well-studied 
original model, has not yet been 
analyzed. 

Previous work, e.g. \cite{Fab03,Al06,De07,MS10}, has shown that the price of anarchy for the \textsc{Sum}-BG and the \textsc{Max}-BG is constant for a wide range of $\alpha$ and in $2^{\mathcal{O}(\sqrt{\log n})}$ in general. For the \textsc{Sum}-(A)SG the best upper bound is in $2^{\mathcal{O}(\sqrt{\log n})}$ as well~\cite{ADHL10,MS12}, whereas the \textsc{Max}-SG has a lower bound of $\Omega(\sqrt{n})$~\cite{ADHL10}. Interestingly, if played on trees, then the \textsc{Sum}-SG and the \textsc{Max}-SG have constant price of anarchy~\cite{ADHL10}, whereas the \textsc{Sum}-ASG and the bounded budget version on trees has price of anarchy in $\Theta(\log n)$~\cite{Ehs11,MS12}. Moreover, it is easy to show that the \textsc{Max}-ASG on trees has price of anarchy in $\Theta(n)$. Thus, we have the desirable property that selfish behavior leads to a relatively small deterioration in social welfare for most of the proposed versions.   

In earlier work~\cite{L11} we studied the game dynamics of the \textsc{Sum}-SG and showed that if the initial network $G_0$ is a tree on $n$ nodes, then the network creation process is guaranteed to converge in $\mathcal{O}(n^3)$ steps. By employing the max cost policy, this process can be sped up significantly to $\mathcal{O}(n)$ steps, which is asymptotically optimal. For the \textsc{Sum}-SG on general networks we showed that there exists a best response cycle, which implies that the \textsc{Sum}-SG on arbitrary initial networks is not a FIPG. 

Very recently, Cord-Landwehr et al.~\cite{CHKS12} studied a variant of the \textsc{Max}-SG, where agents have communication interests, and showed that this variant admits a best response cycle on a tree network as initial network. Hence the restricted-interest variant of the \textsc{Max}-SG is not a FIPG -- even on trees. 

Brandes et al.~\cite{BHN08} were the first to observe that the \textsc{Sum}-BG is not a FIPG and they prove this by providing a better response cycle. Very recently, Bil{\`o} et al.~\cite{Bilo12} gave a better response cycle for the \textsc{Max}-BG which implies the same statement for this version. Note, that both proofs contain agents who perform a sub-optimal move at some step in the better response cycle. Hence, these two results do not address the convergence behavior if agents play optimally. 
\subsection{Our Contribution}
In this work, we study Network Creation Games, as proposed by Fabrikant et al.~\cite{Fab03}, and several natural variants of this model from a new perspective. Instead of analyzing properties of equilibrium states, we apply a more constructive point of view by asking if and how fast such desirable states can be found by selfish agents. For this, we turn the original model and its variants, which are originally formulated as one-shot simultaneous-move games, into more algorithmic models, where moves are performed sequentially.


For the \textsc{Max} Swap Game on trees, we show that the process must converge in $\mathcal{O}(n^3)$ steps, where $n$ is the number of agents. Furthermore, by introducing a natural way of coordination we obtain a significant speed-up to $\Theta(n\log n)$ steps, which is almost optimal. We show that these results, combined with results from our earlier work~\cite{L11}, give the same bounds for the Asymmetric Swap Game on trees in both the \textsc{Sum}- and the \textsc{Max}-version. 

These positive results for initial networks which are trees are contrasted by several strong negative results on general networks. We show that the \textsc{Max}-SG, the \textsc{Sum}-ASG and the \textsc{Max}-ASG on general networks are \emph{not} guaranteed to converge if agents repeatedly perform best possible improving moves and, even worse, that \emph{no} move policy can enforce convergence. We show that these games are not in FIPG, which implies that there cannot exist a generalized ordinal potential function which ``guides'' the way towards an equilibrium state. For the \textsc{Sum}-ASG we show the even stronger negative result that it can happen that \emph{no} sequence of best response moves may enforce convergence, that is, the \textsc{Sum}-ASG is not even weakly acyclic under best response. If not all possible edges can be created, that is if we have a non-complete \emph{host graph} \cite{De09,Bilo12}, then we show that the \textsc{Sum}-ASG and the \textsc{Max}-ASG on non-tree networks is not weakly 
acyclic. Moreover, we map the 
boundary between convergence and non-convergence in ASGs and show the surprising result that cyclic behavior can already occur in $n$-vertex networks which have $n$ edges. That is, even one non-tree edge suffices to completely change the dynamic behavior of these games. In our constructions we have that every agent owns exactly one edge, which is equivalent to the uniform-budget case introduced by Ehsani et al.~\cite{Ehs11}. In their paper~\cite{Ehs11} the authors raise the open problem of determining the convergence speed for the bounded-budget version. Thus, our results answer this open problem -- even for the simplest version of these games -- in the negative, since we show that no convergence guarantee exists. 

We provide best response cycles for all versions of the Buy Game, which implies that these games have no convergence guarantee -- even if agents have the computational resources to repeatedly compute best response strategies. To the best of our knowledge, the existence of \emph{best} response cycles for all these versions was not known before. 
Furthermore, we investigate the version where bilateral consent is needed for edge-creation and where the edge-cost is shared equally among its endpoints. We show that this version exhibits a similar undesirable dynamic behavior as the unilateral version. Quite surprisingly, we can show an even stronger negative result in the \textsc{Sum}-version which implies the counter-intuitive statement that cost-sharing may lead to worse dynamic behavior. Our findings nicely contrast a result of Corbo and Parkes~\cite{CP05} who show guaranteed convergence if agents repeatedly play best response strategies against \emph{perturbations} of the other agents' strategies. We show, that these perturbations are necessary for achieving convergence.

Finally, we present a careful empirical study of the convergence time in the ASG and in the GBG. Interestingly, our simulations show that our negative theoretical results seem to be confined to a small set of pathological instances. Even more interesting may be that our simulations show a remarkably fast convergence towards stable networks in $\mathcal{O}(n)$ steps, where $n$ is the number of agents. This indicates that despite our negative results distributed local search may be a suitable method for selfish agents for collectively finding equilibrium networks.

\section{Dynamics in \textsc{Max} Swap Games}\label{sec_maxsg}
In this section we focus on the game dynamics of the \textsc{Max}-SG. Interestingly, we obtain results which are very similar to the results shown in our earlier work~\cite{L11} but we need entirely different techniques to derive them. Omitted proofs can be found in the Appendix. 
\subsection{Dynamics on Trees}\label{sec_maxsg_tree}
We will analyze the network creation process in the \textsc{Max}-SG when the initial network is a tree. We prove that this process has the following desirable property:
\begin{theorem}\label{thm_tree_pot}
The \textsc{Max}-SG on trees is guaranteed to converge in $\mathcal{O}(n^3)$ steps to a stable network. That is, the \textsc{Max}-SG on trees is a poly-FIPG. 
\end{theorem}
Before proving Theorem~\ref{thm_tree_pot}, we analyze the impact of a single edge-swap. Let $T = (V,E)$ be a tree on $n$ vertices and let agent $v$ be unhappy in network~$T$. Assume that agent $v$ can decrease her cost by performing the edge-swap $vu$ to $vw$, for some $u,w \in V$. This swap transforms $T$ into the new network $T' = (V,(E\setminus\{vu\})\cup\{vw\})$. Let $c_T(v) = \max_{x\in V(T)}d_T(v,x)$ denote agent $v$'s cost in the network $T$. Let $c_{T'}(u)$ denote her respective cost in $T'$. Let $A$ denote the tree of $T'' = (V,E\setminus\{vu\})$ which contains $v$ and let $B$ be the tree of $T''$ which contains $u$ and $w$. It is easy to see, that we have $d_T(x,y) = d_{T'}(x,y)$, if $x,y \in V(A)$ or if $x,y \in V(B)$. 
\begin{lemma}\label{lem_tree_subtree}
For all $x\in V(A)$ there is no $y\in V(A)$ such that $c_T(x) = d_T(x,y)$. 
\end{lemma}
\begin{proof}
By assumption, agent $v$ can decrease her cost by swapping the edge $vu$ to edge $vw$, where $u,w \in V(B)$. We have that $d_T(x,v) < c_T(v)$, for all $x \in V(A)$, since otherwise this swap would not change agent $v$'s cost. It follows that for arbitrary $x,y \in V(A)$ we have $d_T(x,y) \leq d_T(x,v) + d_T(v,y) < d_T(x,v) + c_T(v).$ Let $z \in V(B)$ be a vertex having maximum distance to $v$ in $T$, that is, $c_T(v) = d_T(v,z)$. The above implies that $d_T(x,y) < d_T(x,z) = c_T(x)$, for all $x,y \in V(A)$.
\end{proof}
\noindent Lemma~\ref{lem_tree_subtree} directly implies the following statement:
\begin{corollary}\label{cor_tree_subtree}
For all $x \in V(A)$, we have $c_T(x) > c_{T'}(x)$.
\end{corollary}
\noindent Hence, we have that agent $v$'s improving move decreases the cost for all agents in $V(A)$. For agents in $V(B)$ this may not be true: The cost of an agent $y \in V(B)$ can increase by agent $v$'s move. Interestingly, the next result guarantees that such an increase cannot be arbitrarily high.
\begin{lemma}\label{lem_tree_costincrease}
Let $x \in V(A)$ and $y \in V(B)$ such that $d_{T'}(x,y) = c_{T'}(y)$. It holds that $c_T(x)>c_{T'}(y)$. 
\end{lemma}
\begin{proof}
In tree $T$ we have $c_T(x) = d_T(x,v) + d_T(u,z) + 1$. Furthermore, in tree $T'$ we have $c_{T'}(y) = d_{T'}(x,v) + d_{T'}(w,y) + 1$. Since $c_T(v)>c_{T'}(v)$, it follows that $d_T(w,y) < d_T(u,z)$, where $z\in V(B)$ is a vertex having maximum distance to $v$ in $T$. Hence, we have $c_T(x) - c_{T'}(y) = d_T(u,z) - d_T(w,y) > 0.$
\end{proof}
\noindent Towards a generalized ordinal potential function we will need the following:
\begin{definition}[Sorted Cost Vector and Center-Vertex]
Let $G$ be any network on $n$ vertices. The \emph{sorted cost vector} of $G$ is $\overrightarrow{c_G} = (\gamma_G^1,\dots,\gamma_G^n)$, where $\gamma_G^i$ is the cost of the agent, who has the $i$-th highest cost in the network $G$. An agent having cost $\gamma_G^n$ is called \emph{center-vertex} of $G$.
\end{definition}
\begin{lemma}\label{lem_tree_potfunct}
Let $T$ be any tree on $n$ vertices. The sorted cost vector of $T$ induces a generalized ordinal potential function for the \textsc{Max}-SG on $T$. 
\end{lemma}
\begin{proof}
Let $v$ be any agent in $T$, who performs an edge-swap which strictly decreases her cost and let $T'$ denote the network after agent $v$'s swap. We show that $c_T(v) - c_{T'}(v) > 0$ implies $\overrightarrow{c_T} >_{\text{lex}} \overrightarrow{c_{T'}},$ where $>_\text{lex}$ is the lexicographic order on $\mathbb{N}^n$. The existence of a generalized ordinal potential function then follows by mapping the lexicographic order on $\mathbb{N}^n$ to an isomorphic order on $\mathbb{R}$.   

Let the subtrees $A$ and $B$ be defined as above and let $c_T(v) - c_{T'}(v) > 0$. By Lemma~\ref{lem_tree_subtree} and Lemma~\ref{lem_tree_costincrease}, we know that there is an agent $x \in V(A)$ such that $c_T(x) > c_{T'}(y)$, for all $y \in V(B)$. By Lemma~\ref{lem_tree_subtree} and Corollary~\ref{cor_tree_subtree}, we have that $c_T(x) > c_{T'}(x)$, which implies that $\overrightarrow{c_T} >_{\text{lex}} \overrightarrow{c_{T'}}$. 
\end{proof}
In the following, a special type of paths in the network will be important.
\begin{definition}[Longest Path]
Let $G$ be any connected network. Let $v$ be any agent in $G$ having cost $c_G(v) = k$. Any simple path in $G$, which starts at $v$ and has length $k$ is called a \emph{longest path of agent $v$}.
\end{definition}
\noindent As we will see, center-vertices and longest paths are closely related. 
\begin{lemma}\label{lem_tree_center}
Let $T$ be any connected tree and let $v^*$ be a center-vertex of $T$. Vertex $v^*$ must lie on all longest paths of all agents in $V(T)$. 
\end{lemma}
\begin{proof}
Let $P_{xy}$ denote the path from vertex $x$ to vertex $y$ in $T$. We assume towards a contradiction that there are two vertices $v,w \in V(T)$, where $c_T(v) = d_T(v,w)$, and that $v^* \notin V(P_{vw})$. Let $z \in V(T)$ be the only shared vertex of the three paths $P_{vv^*}, P_{wv^*}, P_{vw}$. We have $d_T(v,z) < d_T(v,v^*)\leq c_T(v^*)$ and $d_T(w,z) < d_T(w,v^*) \leq c_T(v^*)$. We show that $c_T(z) < c_T(v^*)$, which is a contradiction to $v^*$ being a center-vertex in $T$.

Assume that there is a vertex $u \in V(T)$ with $d_T(u,z) \geq c_T(v^*)$. It follows that $V(P_{vz}) \cap V(P_{zu}) = \{z\}$, since otherwise $d_T(v^*,u) = d_T(v^*,z) + d_T(z,u) > c_T(v^*)$. But now, since $d_T(z,w) < c_T(v^*) \leq d_T(z,u)$, we have $d_T(v,u) > c_T(v)$, which clearly is a contradiction. Hence, we have $d_T(z,u) < c_T(v^*)$, for all $u \in V(T)$, which implies that $c_T(z) < c_T(v^*)$.
\end{proof}
Lemma~\ref{lem_tree_center}, leads to the following observation.
\begin{observation}\label{obs_cost_vector}
Let $G$ be any connected network on $n$ nodes and let $\overrightarrow{c_G} = (\gamma_G^1,\dots,\gamma_G^n)$ be its sorted cost vector. We have $\gamma_G^1 = \gamma_G^2$ and $\gamma_G^n = \left\lceil \frac{\gamma_G^1}{2} \right\rceil$.
\end{observation}
Now we are ready to provide the key property which will help us upper bound the convergence time.
\begin{lemma}\label{lem_tree_diameter}
Let $T = (V,E)$ be a connected tree on $n$ vertices having dia\-meter $D\geq 4$. After at most $\tfrac{nD-D^2}{2}$ moves of the \textsc{Max}-SG on $T$ one agent must perform a move which decreases the diameter. 
\end{lemma}
\begin{proof}
Let $v,w \in V$ such that $d_T(v,w) = D \geq 4$ and let $P_{vw}$ be the path from $v$ to $w$ in $T$. Clearly, if no agent in $V(P_{vw})$ makes an improving move, then the diameter of the network does not change. On the other hand, if the path $P_{vw}$ is the unique path in $T$ having length $D$, then any improving move of an agent in $V(P_{vw})$ must decrease the diameter by at least $1$. The network creation process starts from a connected tree having diameter $D\geq 4$ and, by Lemma~\ref{lem_tree_potfunct}, must converge to a stable tree in a finite number of steps. Moreover, Lemma~\ref{lem_tree_potfunct} guarantees that the diameter of the network cannot increase in any step of the process. It was shown by Alon et al.~\cite{ADHL10} that any stable tree has diameter at most~$3$. Thus, after a finite number of steps the diameter of the network must strictly decrease, that is, on all paths of length $D$ some agent must have performed an improving move which reduced the length of the respective path. We fix 
the path $P_{vw}$ to be the path of length $D$ in the network which survives longest in this process.       

It follows, that there are $|V\setminus V(P_{vw})| = n - (D+1)$ agents which can perform improving moves without decreasing the diameter. We know from Observation~\ref{obs_cost_vector} and Lemma~\ref{lem_tree_center} that each one of those $n - (D+1)$ agents can decrease her cost to at most $\left\lceil \frac{D}{2}\right\rceil + 1$ and has to decrease her cost by at least $1$ for each edge-swap. We show that an edge-swap of such an agent does not increase the cost of any other agent and use the minimum possible cost decrease per step to 
conclude 
the desired bound.  

Let $u \in V(T)\setminus V(P_{vw})$ be an agent who decreases her cost by swapping the edge $ux$ to $uy$ and let $T'$ be the tree after this edge-swap. Let $a,b \in V(T)$ be arbitrary agents. Clearly, if $\{u,y\} \not\subseteq V(P_{ab})$ in $T'$, then $d_T(a,b) = d_{T'}(a,b)$. Let $A$ be the tree of $T'' = (V,E\setminus\{uy\})$ which contains $u$ and let $B$ be the tree of $T''$ which contains $y$. W.l.o.g. let $a \in V(A)$ and $b \in V(B)$. By Corollary~\ref{cor_tree_subtree}, we have $c_T(z) > c_{T'}(z)$ for all $z \in V(A)$ and it follows that $V(A) \cap V(P_{vw}) = \emptyset$. Hence, it remains to analyze the change in cost of all agents in $V(B)$. 

If no vertex on the path $P_{ab}$ is a center-vertex in $T'$, then, by Lemma~\ref{lem_tree_center}, we have that $d_{T'}(a,b) < c_{T'}(b)$. It follows that every longest path of agent $b$ in $T'$ lies entirely in subtree $B$ which implies that $c_{T'}(b) \leq c_T(b)$.

If there is a center-vertex of $T'$ on the path $P_{ab}$ in $T'$, then let $v^*$ be the last such vertex on this path. We have assumed that the diameters of $T'$ and $T$ are equal, which implies that $P_{vw}$ is a longest path of agent $v$ in $T'$. Since, by Lemma~\ref{lem_tree_center}, any center-vertex of $T'$ must lie on all longest paths, it follows that $v^*$ is on the path $P_{vw}$ and we have $v^* \in V(B)$. W.l.o.g. let $d_{T'}(v,b) \geq d_{T'}(w,b)$. We have $d_{T'}(a,b) = d_{T'}(a,v^*) + d_{T'}(v^*,b) \leq d_{T'}(v,v^*) + d_{T'}(v^*,b)$. Hence, we have $d_{T'}(a,b) \leq c_{T'}(b)$. Since the path $P_{bv}$ is in subtree $B$, we have $c_{T'}(b) \leq c_T(b)$.

Now we can easily conclude the upper bound on the number of moves which do not decrease the diameter of $T$. Each of the $n - (D+1)$ agents with cost at most $D$ may decrease their cost to $\left\lceil\frac{D}{2}\right\rceil + 1$. If we assume a decrease of $1$ per step, then this yields the following bound: 
 $$(n - (D+1)) \left(D - \left(\left\lceil \frac{D}{2}\right\rceil +1  \right) \right) < (n-D)\frac{D}{2} = \frac{nD - D^2}{2}.$$
\end{proof}

\begin{proofof}{Theorem~\ref{thm_tree_pot}}
By Lemma~\ref{lem_tree_potfunct}, we know there exists a generalized ordinal potential function for the \textsc{Max}-SG on trees. Hence, we know that this game is a FIPG and we are left to bound the maximum number of improving moves needed for convergence. It was already shown by Alon et al.~\cite{ADHL10}, that the only stable trees of the \textsc{Max}-SG on trees are stars or double-stars. Hence, the process must stop at the latest when diameter $2$ is reached. 

Let $N_n(T)$ denote the maximum number of moves needed for convergence in the \textsc{Max}-SG on the $n$-vertex tree $T$. Let $D(T)$ be the diameter of $T$. 
Let $D_{i,n}$ denote the maximum number of steps needed to decrease the diameter of any $n$-vertex tree having diameter $i$ by at least $1$. Hence, we have $$N_n(T) \leq \sum_{i=3}^{D(T)} D_{i,n} \leq \sum_{i=3}^{n-1} D_{i,n},$$ since the maximum diameter of a $n$-vertex tree is $n-1$. By applying Lemma~\ref{lem_tree_diameter} and adding the steps which actually decrease the diameter, this yields $$N_n(T) \leq \sum_{i=3}^{n-1} D_{i,n} < \sum_{i=3}^{n-1} \left(\frac{ni-i^2}{2} + 1 \right) < n + \frac{n}{2} \left(\sum_{i=1}^{n}i\right) - \frac{1}{2}\left(\sum_{i=1}^n i^2\right) \in \mathcal{O}(n^3).$$
\end{proofof}
The following result shows that we can speed up the convergence time by employing a very natural move policy. The speed-up is close to optimal, since it is easy to see that there are instances in which $\Omega(n)$ steps are necessary. The proof can be found in the Appendix. 
\begin{theorem}\label{thm_mcbrd_tree}
The \textsc{Max}-SG on trees with the max cost policy converges in $\Theta(n\log n)$ moves.
\end{theorem}
\begin{proof}
We prove Theorem~\ref{thm_mcbrd_tree}, by proving the lower and the upper bound separately, starting with the former. Since we analyze the max cost policy, we need two additional observations. 
\begin{observation}\label{obs_tree_leaf}
An agent having maximum cost in a tree $T$ must be a leaf of $T$.
\end{observation}
\begin{observation}\label{obs_tree_br}
Let $u$ be an unhappy agent in $T = (V,E)$ and let $u$ be a leaf of $T$ and let $v$ be $u$'s unique neighbor. Let $B$ be the tree of $T' = (V,E\setminus\{uv\})$ which contains $v$. The edge-swap $uv$ to $uw$, for some $w\in V(B)$ is a best possible move for agent $u$ if $w$ is a center-vertex of~$B$.   
\end{observation}

\begin{lemma}\label{lem_tree_mc_lower}
There is a tree $T$ on $n$ vertices where the \textsc{Max}-SG on $T$ with the max cost policy needs $\Omega(n\log n)$ moves for convergence.  
\end{lemma}
\begin{proof}
We consider the path on $n$-vertices $P_n = v_1v_2\dots v_n$ of length $n-1$. We apply the max cost policy and for breaking ties we will always choose the vertex having the smallest index among all vertices having maximum cost. If a maximum cost vertex has more than one best response move, then we choose the edge-swap towards the new neighbor having the smaller index. With these assumptions and with Observation~\ref{obs_tree_leaf} and Observation~\ref{obs_tree_br}, we have that the center-vertex having the smallest index will ``shift'' towards a higher index, from $v_{\lceil n/2 \rceil}$ to $v_{n-2}$. Finally, agent $v_n$ is the unique agent having maximum cost and her move transforms the tree to a star. See Fig.~\ref{fig:mc_tree_example} for an illustration for $n=9$.
\begin{figure}[!h]
  \centering
  \includegraphics[width=\textwidth]{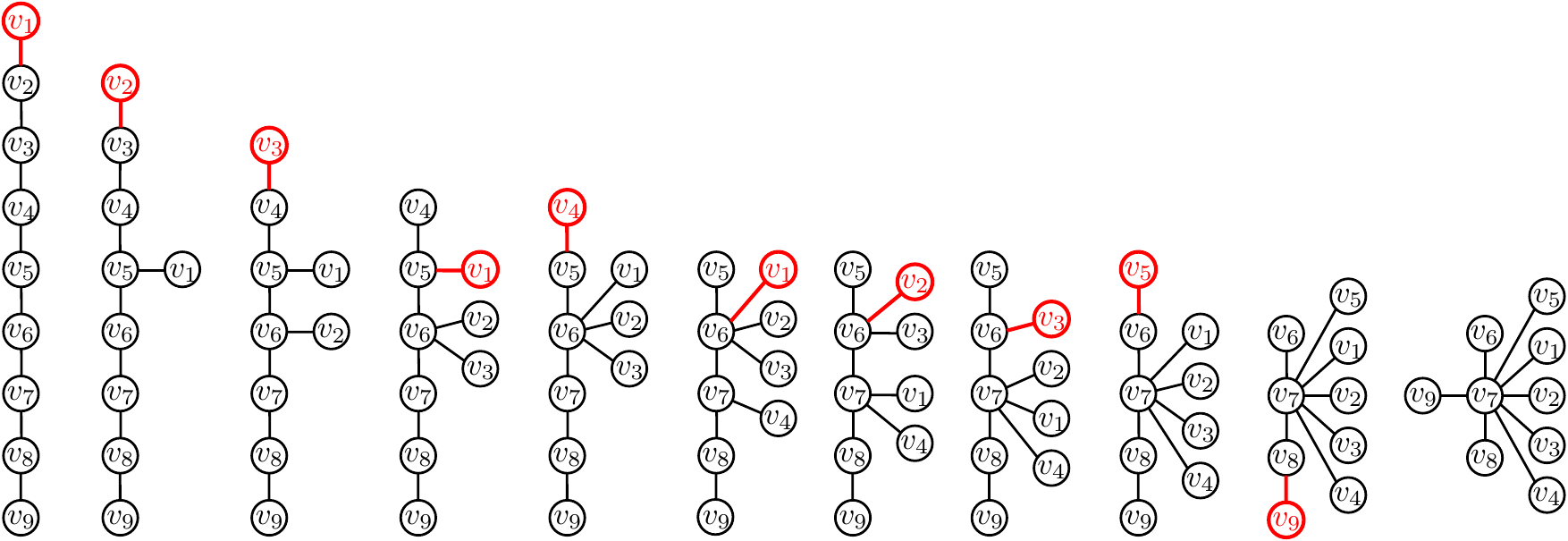}
  \caption{The convergence process for with $n=9$ in the \textsc{Max}-SG on $P_n$.}
  \label{fig:mc_tree_example}
\end{figure}

\noindent We start by analyzing the change in costs of agent $v_1$. Clearly, $c_0 = c_{P_n}(v_1) = n-1$. By Observation~\ref{obs_tree_br}, we know that $v_1$'s best swap connects to the minimum index center-vertex of the tree without vertex $v_1$. Hence after the best move of $v_1$ this agent has cost $c_1 = \left\lceil \frac{c_0-1}{2} \right\rceil + 1 > \frac{c_0}{2}$.  When $v_1$ is chosen to move again, her cost can possibly decrease to $\left\lceil\frac{c_1 - 1}{2}\right\rceil +1 > \frac{c_0}{4}$. After the $i$-th move of $v_1$ her cost is at least $\left\lceil\frac{c_{i-1}-1}{2}\right\rceil+1 > \frac{c_0}{2^i}$. Thus, the max cost policy allows agent $v_1$ to move at least $\log \frac{c_0}{3}$ times until she is connected to vertex $v_{n-2}$, the center of the final star, where she has cost~$3$.

The above implies, that the number of moves of every agent allowed by the max cost policy only depends on the cost of that agent when she first becomes a maximum cost agent. Moreover, since all moving agents are leaves, no move of an agent increases the cost of any other agent. By construction, the cost of every moving agent is determined by her distance towards vertex $v_n$. Since agent $v_n$ does not move until in the last step of the process, we have that a move of agent $v_i$ does not change the cost of any other agent $v_j \neq v_n$ who moves after $v_i$. It follows, that we can simply add up the respective lower bounds on the number of moves of all players, depending on the cost when they first become maximum cost agents. It is easy to see, that agent $v_i$ becomes a maximum cost agent, when the maximum cost is $n-i$. Let $M(P_n)$ denote the number of moves of the \textsc{Max}-SG on $P_n$ with the max cost policy and the above tie-breaking rules. This yields $$M(P_n) > \sum_{c_0 = n-1}^4 \log \frac{c_0}
{
3} \in \Omega(n\log n).$$ 
\end{proof}
\begin{lemma}\label{lem_tree_mc_upper}
The \textsc{Max}-SG on a $n$-vertex tree $T$ with the max cost policy needs $\mathcal{O}(n\log n)$ moves to converge to a stable tree.
\end{lemma}
\begin{proof}
Consider any tree $T$ on $n$ vertices. By Observation~\ref{obs_tree_leaf}, we know that only leaf-agents are allowed to move by the max cost policy, which implies that no move of any agent increases the cost of any other agent. Observation~\ref{obs_tree_br} guarantees that the best possible move of a leaf-agent $u$ having maximum cost $c$ decreases agent $u$'s cost to at most $\left\lceil\frac{c}{2}\right\rceil+1$. Hence, after $\mathcal{O}(\log n)$ moves of agent $u$ her cost must be at most $3$. If the tree converges to a star, then agent $u$ may move one more time. If we sum up over all $n$ agents, then we have that after $\mathcal{O}(n\log n)$ moves the tree must be stable.
\end{proof}
This concludes the proof.
\end{proof}

\subsection{Dynamics on General Networks}\label{sec_maxsg_nontree}
In this section we show that allowing cycles in the initial network completely changes the dynamic behavior of the \textsc{Max}-SG. The proof can be found in the Appendix.
\begin{theorem}\label{thm_maxsg_brcycle}
The \textsc{Max}-SG on general networks admits best response cycles. Moreover, no move policy can enforce convergence. The first result holds even if agents are allowed to perform multi-swaps.
\end{theorem}
\begin{proof}
We prove the theorem by showing that there exists an initial network which induces a best response cycle and where in every step of the cycle exactly one agent is unhappy. The existence of the best response cycle shows that the \textsc{Max}-SG on this instance does not have the finite improvement property. The fact that in every step of the cycle exactly one agent is unhappy shows that no move policy can avoid that cyclic behavior. In every step, swapping one edge suffices to achieve the best possible cost decrease for the moving agent. Hence, there exists a best response cycle even if agents are allowed to perform multi-swaps. However, note that with multi-swaps it is no longer true that there is only one unhappy agent in every step.
\begin{figure}[!h]
  \centering
  \includegraphics[width=13cm]{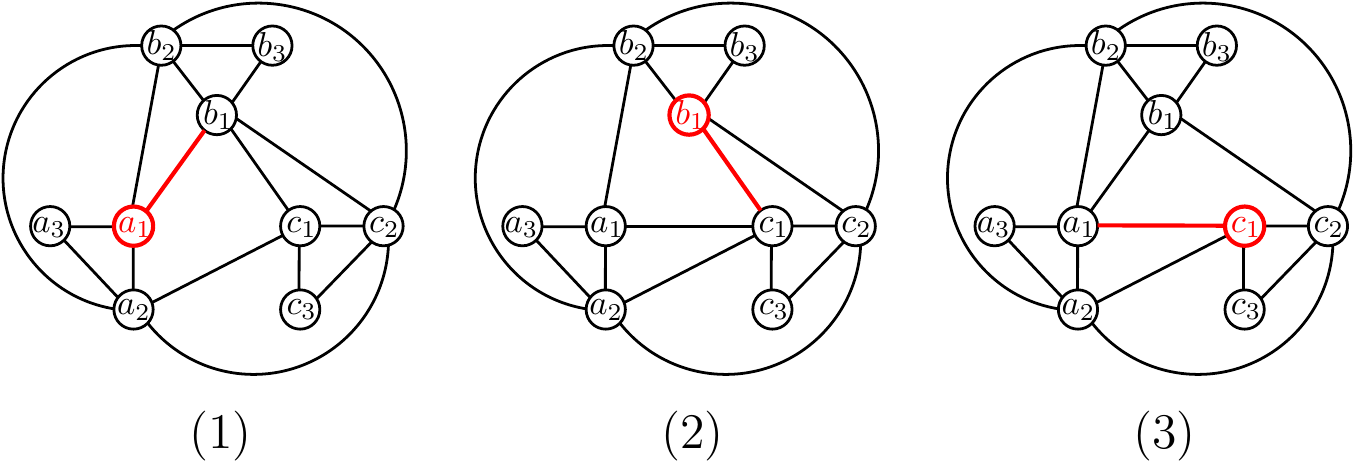}
  \caption{The steps of a best response cycle for the \textsc{Max}-SG on general networks.}
  \label{fig:maxsgbr}
\end{figure}

Consider the initial network $G_1$ which is depicted in Fig.~\ref{fig:maxsgbr}~(1). Note, that only agents $a_1,a_3,b_3$ and $c_3$ have cost $3$ while all other agents have cost $2$. Clearly, agents having cost $2$ cannot improve on their current situation. Agents $a_3,b_3,c_3$ cannot perform an improving move, since all of them have exactly two vertices in distance $3$ and there is no vertex which is a neighbor of both of them. This leaves agent $a_1$ as the only possible candidate for an improving edge-swap. A best possible move for agent $a_1$ is the swap $a_1b_1$ to $a_1c_1$, which yields a distance decrease of $1$ which is clearly optimal. This swap transforms $G_1$ into $G_2$, which is depicted in Fig.~\ref{fig:maxsgbr}~(2). Observe, that $G_2$ is isomorphic to $G_1$, with agent $b_1$ facing exactly the same situation as agent $a_1$ in $G_1$. Agent $b_1$ has the swap $b_1c_1$ to $b_1a_1$ as best response move and we end up in network $G_3$, shown in Fig.~\ref{fig:maxsgbr}~(3). Again, $G_3$ is 
isomorphic to $G_1$, now with agent $c_1$ being the unique unhappy agent. Agent $c_1$'s best possible swap transforms $G_3$ back into $G_1$. 
\end{proof}


\section{Dynamics in Asymmetric Swap Games}\label{sec_asg}
In this section we consider the \textsc{Sum}-ASG and the \textsc{Max}-ASG. Note, that now we assume that each edge has an owner and only this owner is allowed to swap the edge. We show that we can directly transfer the results from above and from \cite{L11} to the asymmetric version if the initial network is a tree. On general networks we show even stronger negative results. Omitted proofs can be found in the Appendix.

Observe, that the instance used in the proof of Theorem~\ref{thm_maxsg_brcycle} and the corresponding instance in \cite{L11} show that best response cycles in the Swap Game are not necessarily best response cycles in the Asymmetric Swap Game. We will show the rather counter-intuitive result that this holds true for the other direction as well. 
\subsection{Asymmetric Swap Games on Trees}\label{sec_asg_tree}
The results in this section follow from the respective theorems in \cite{L11} and from the results in Section~\ref{sec_maxsg_tree} and are therefore stated as corollaries.
\begin{corollary}\label{cor_polyfipg}
 The \textsc{Sum}-ASG and the \textsc{Max}-ASG on $n$-vertex trees are both a poly-FIPG and both must converge to a stable tree in $\mathcal{O}(n^3)$ steps.
\end{corollary}
\begin{proof}
It was shown in \cite{L11} that the \textsc{Sum}-SG on trees is an ordinal potential game, where the social cost, which is the sum of all agent's costs, serves as ordinal potential function. Furthermore, it was shown that the \textsc{Sum}-SG on $n$-vertex trees must converge in $\mathcal{O}(n^3)$ steps. Note, that ordinal potential games are a subclass of FIPG~\cite{MS96}.   

The only difference in the \emph{asymmetric} version of this game is that edges have owners and only the respective owner is allowed to swap an edge. Clearly, since any improving swap decreases the value of the (generalized) potential function, this is independent of the edge-ownership. Furthermore, with edges having owners, we have that in each network less moves are possible and every moving agent has, compared with the Swap Game, at most the same number of admissible strategies in any step. Thus, the convergence process cannot be slower. The results from \cite{L11} and Theorem~\ref{thm_tree_pot} then yields the desired statement.
\end{proof}
\begin{corollary}\label{cor_mcbrd}
Using the max cost policy and assuming a $n$-vertex tree as initial network, we have that\vspace*{-0.2cm}
\begin{itemize}
 \item[$\bullet$] the \textsc{Sum}-ASG converges in $\max\{0,n-3\}$ steps, if $n$ is even and in $\max\{0,n+\lceil n/2 \rceil -5\}$ steps, if $n$ is odd. Moreover, both bounds are tight and asymptotically optimal.
 \item[$\bullet$] the \textsc{Max}-ASG converges in $\Theta(n\log n)$ steps.
\end{itemize}
\end{corollary}
\begin{proof}
 We can carry over the results from \cite{L11} and Section~\ref{sec_maxsg_tree} about speeding up the convergence process by a suitable move policy. The reason for this is that in all used lower bound constructions it holds that whenever an edge is swapped more than once, then it is the same incident agent who moves again. Hence, we can assign the edge-ownership to this agent and get the same lower bounds in the asymmetric version. The upper bounds carry over trivially, since agents cannot have more admissible new strategies in any step in the asymmetric version compared to the version without edge-owners.
\end{proof}

\subsection{Asymmetric Swap Games on General Graphs}\label{sec_asg_nontree}
If we move from trees to general initial networks, we get a very strong negative result for the \textsc{Sum}-ASG: There is no hope to enforce convergence if agents stick to playing best responses even if multi-swaps are allowed. 
\begin{theorem}\label{thm_sumasymswap}
  The \textsc{Sum}-ASG on general networks is not weakly acyclic under best response. Moreover, this result holds true even if agents can swap multiple edges in one step. 
\end{theorem}
\begin{proof}
 We give a network which induces a best response cycle. Additionally, we show that in each step of this cycle exactly one agent can decrease her cost by swapping an edge and that the best possible swap for this agent is unique in every step. Furthermore, we show that the moving agent cannot outperform the best possible single-swap by a multi-swap. This implies that if agents stick to best response moves then \emph{no} best response dynamic can enforce convergence to a stable network and allowing multi-swaps does not alter this result.
 \begin{figure}[!h]
  \centering
  \includegraphics[width=\textwidth]{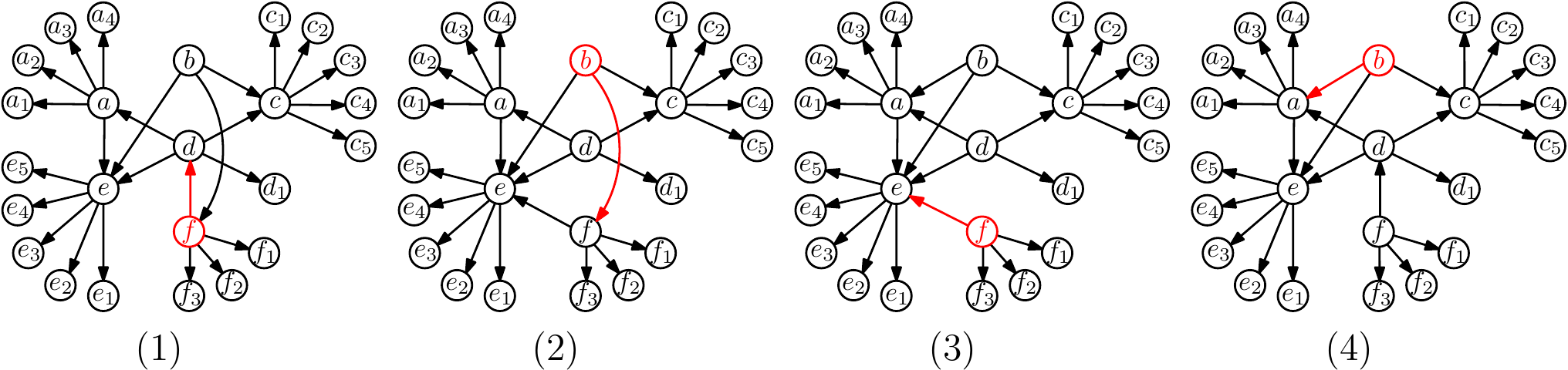}
  \caption{The steps of a best response cycle for the \textsc{Sum}-ASG on general networks. Note, that edge directions indicate edge-ownership. All edges are two-way.}
  \label{fig:asymswapbr}
\end{figure}

 The best response cycle consists of the networks $G_1, G_2, G_3$ and $G_4$ given in Fig.~\ref{fig:asymswapbr}. We begin with showing that in $G_1,\dots,G_4$ all agents, except agent $b$ and agent $f$, cannot perform an improving strategy change even if they are allowed to swap multiple edges in one step.  

 In $G_1,\dots,G_4$ all leaf-agents do not own any edges and the agents $c$ and $e$ cannot swap an edge since otherwise the network becomes disconnected. For the same reason, agent~$d$ cannot move the edge towards $d_1$. Agent~$d$ owns three other edges, but they are optimally placed since they are connected to the vertices having the most leaf-neighbors. It follows, that agent~$d$ cannot decrease her cost by swapping one edge or by performing a multi-swap. Note, that this holds true for all networks $G_1,\dots,G_4$, although the networks change slightly. Agent~$a$ cannot move her edges towards $a_i$, for $1\leq i \leq 4$. On the other hand, it is easy to see that agent~$a$'s edge towards vertex $e$ cannot be swapped to obtain a strict cost decrease since the most promising choice, which is vertex $c$, yields the same cost in $G_1$ and $G_4$ and even higher cost in $G_2$ and $G_3$. Trivially, no multi-swap is possible for agent~$a$.
 
 Now, we consider agent $b$ and agent $f$. First of all, observe that in $G_1,\dots,G_4$ agent~$f$ owns exactly one edge which is not a bridge. Thus, agent $f$ cannot perform a multi-swap in any step of the best response cycle. Agent~$b$, although owning three edges, is in a similar situation: Her edges to vertex $c$ and $e$ can be considered as fixed, since swapping one or both of them does not yield a cost decrease in $G_1,\dots,G_4$. Hence, agent $b$ and agent $f$ each have one ``free'' edge to operate with. In $G_1$ agent $b$'s edge towards $f$ is placed optimally, since swapping towards $a$ or $d$ does not yield a cost decrease. In $G_3$, agents $b$'s edge towards $a$ is optimal, since swapping towards $d$ or $f$ does not decrease agent $b$'s cost. Analogously, agent~$f$'s edge towards $e$ in $G_2$ and her edge towards $d$ in $G_4$ are optimally placed. 

 Last, but not least, we describe the best response cycle: In $G_1$ agent~$f$ can improve and her unique best possible edge-swap in $G_1$ is the swap from $d$ to $e$, yielding a cost decrease of~$4$. In $G_2$ agent~$b$ has the swap from $f$ to $a$ as unique best improvement which yields a cost decrease of $1$. In $G_3$ have agent~$f$ being unhappy with her strategy and the unique best swap is the one from $e$ to $d$ yielding an improvement of $1$. In $G_4$ it is agent~$b$'s turn again and her unique best swap is from $a$ to $f$ which decreases her cost by $3$. After agent~$b$'s swap in $G_4$ we arrive again at network $G_1$, hence $G_1,\dots,G_4$ is a best response cycle where in each step exactly one agent has a single-swap as unique best possible improvement.
\end{proof}
\begin{remark}
 Note, that the best response cycle presented in the proof of Theorem~\ref{thm_sumasymswap} is not a best response cycle in the \textsc{Sum-SG}. The swap $fb$ to $fe$ of agent $f$ in $G_1$ yields a strictly larger cost decrease than her swap $fd$ to $fe$.
\end{remark}

Compared to Theorem~\ref{thm_sumasymswap}, we show a slightly weaker negative result for the max-version.  
\begin{theorem}\label{thm_maxasg_brcycle}
 The \textsc{Max}-ASG on general networks admits best response cycles. Moreover, no move policy can enforce convergence.
\end{theorem}
\begin{proof}{Theorem~\ref{thm_maxasg_brcycle}}
We show that there exists a best response cycle for the \textsc{Max}-ASG, where no move policy can enforce convergence. Our cycle, shown in Fig.~\ref{fig:maxasgbr}, has six steps $G_1,\dots,G_6$. In $G_3$ and $G_6$ there are two unhappy agents whereas in the other steps there is exactly one unhappy agent. It turns out, that independently which one of the two agent moves in $G_3$ or $G_6$, there is a best response move which leads back to a network in the cycle. This implies, that no move policy may enforce convergence. 
\begin{figure}[!h]
  \centering
  \includegraphics[width=14cm]{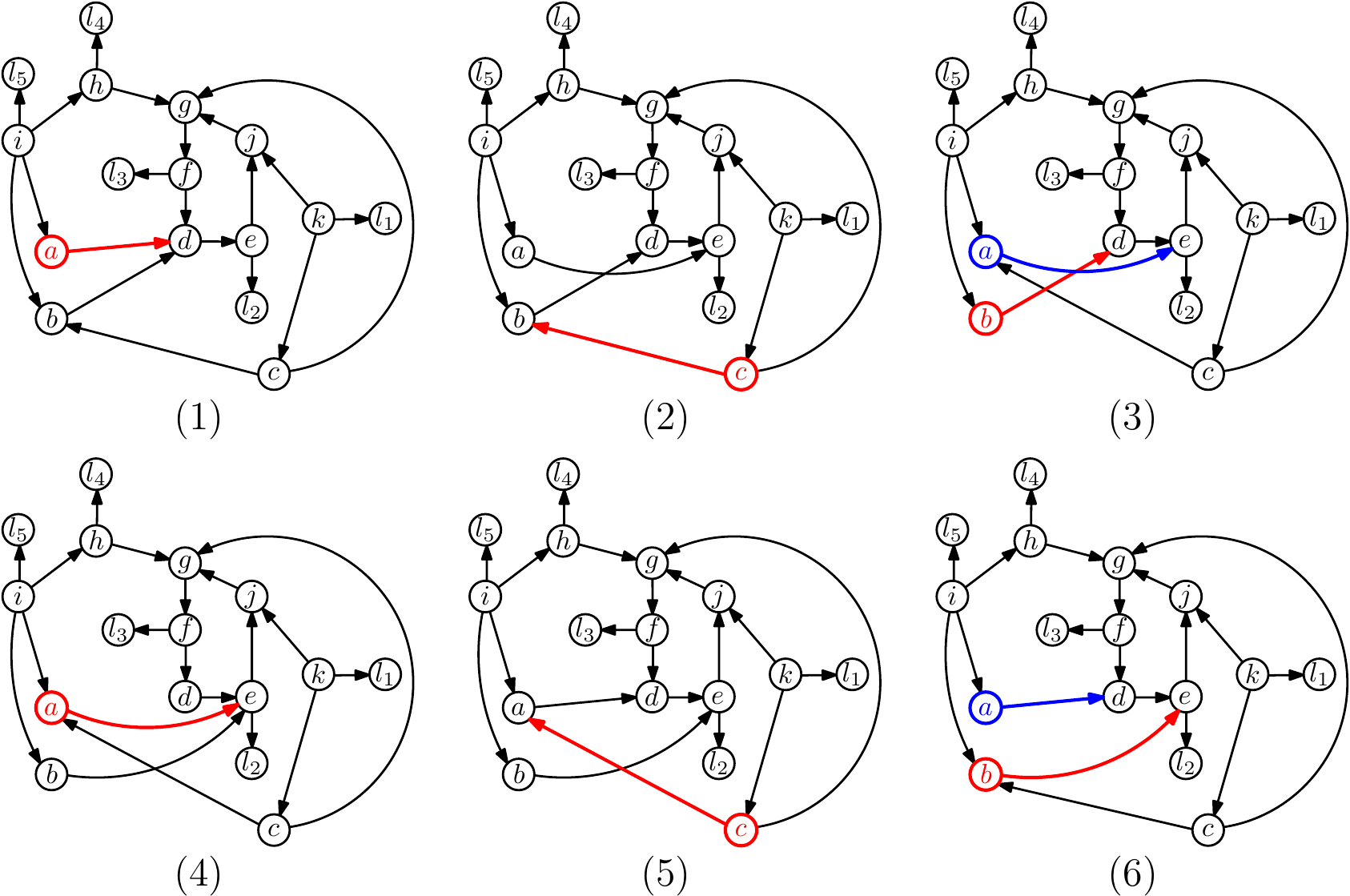}
  \caption{The steps of a best response cycle for the \textsc{Max}-ASG on general networks.}
  \label{fig:maxasgbr}
\end{figure} 

\noindent First of all, note that the networks $G_2$ and $G_5$ are isomorphic. The same holds true for the networks $G_3$ and $G_6$. We start by showing that in the networks $G_1,\dots,G_4$ only the highlighted agents are unhappy. Then we will analyze the best response moves of the unhappy agents in the respective networks. 

Consider any $G_i$, where $1 \leq i \leq 4$. Clearly, no leaf-agent of $G_i$ can perform any swap. Agent $g$ has cost $3$ in $G_i$. We have $d_{G_i}(g,l_1) = d_{G_i}(g,l_5)$ and there is no vertex in $G_i$ which is a neighbor to both $l_1$ and $l_5$. Thus, agent $g$ cannot achieve cost $2$, which implies that agent $g$ does not want to perform a move in~$G_i$. By the same argument, it follows that any agent having cost $3$ must be happy. Hence, we have that agent $b$ in $G_1$ and $G_2$ and agent $c$ in $G_3$ and $G_4$ do not swap. 

Agent $d$, having cost $4$ in $G_i$, cannot perform an improving move in $G_i$, since $d_{G_i}(d,l_1) = d_{G_i}(d,l_4) = 4$ and any such improving move must connect to a vertex which has distance at most $2$ to vertex $l_1$. These vertices are $c,j,k$ and $l_1$ but any of them has distance at least $3$ towards $l_4$ in~$G_i$. Agents $f$ and $i$, both having cost $4$ in $G_i$, each face an analogous situation, since $d_{G_i}(f,l_1) = d_{G_i}(i,l_1) = d_{G_i}(f,l_5) = d_{G_i}(i,l_3) = 4$ and since vertices $c,j,k,l_1$ have distance at least $3$ to $l_5$ or $l_3$ in any $G_i$. 

Agent $h$, having cost $4$ in $G_i$ has both $l_1$ and $l_2$ in distance $4$. The only vertex which has distance at most $2$ to both of them is vertex $j$. But if $h$ swaps towards vertex $j$, then this yields distance $4$ towards towards vertex $l_3$. 

Now we consider agent $e$. Any strategy yielding cost at most $3$ for agent $e$ must connect to $c,j,k$ or $l_1$, since otherwise vertex $l_1$ would be in distance $4$. But, since $d_{G_i}(e,l_4) = 4$ and since $c,j,k,l_1$ all have distance at least $3$ towards $l_4$, no such strategy can exist. For agent $j$, having cost $4$ in $G_i$. the situation is similar. Any strategy having cost at most $3$ for agent $j$ must connect to $a,b,h,i$ or $l_5$, since otherwise $j$'s distance towards $l_5$ would be $4$. But if $j$ swaps away from $g$ to any vertex in $\{a,b,h,i,l_5\}$, then her distance to $l_3$ increases to $4$. 

Agent $k$, having cost $4$ in $G_i$, has vertex $l_4$ and $l_5$ in distance $4$. Thus, to achieve as cost of at most $3$, agent $k$ must connect to vertex $h$ or $i$. Both $h$ and $i$ have distance at least $3$ towards $l_3$. Since $d_{G_i}(k,l_3) = 4$, it follows, that agent $k$ cannot perform an improving move. 

Now, we are left with agent $c$ in $G_1$, agent $a$ in $G_2$ and agent $b$ in $G_4$, all having cost $4$. We have $d_{G_1}(c,l_2) = 4$. Thus, any strategy which yields cost at most $3$ must connect to $d,e,j$ or $l_2$. If $c$ swaps away from $g$, then her distance to $l_4$ increases to $4$. If $c$ swaps away from $b$, then her distance to $l_5$ increases to $4$. Hence, agent $c$ cannot swap an edge to decrease her cost in $G_1$. For agent $a$ in $G_2$ and agent $b$ in $G_4$ the situation is similar. We have $d_{G_2}(a,l_1) = d_{G_2}(a,l_3) = d_{G_4}(b,l_1) = d_{G_4}(b,l_3) = 4$. Both agents cannot decrease their cost by swapping an edge in the respective network, since all vertices which have distance at most $2$ to $l_1$ have distance at least $3$ to $l_3$. 

Finally, we analyze the best response moves of all unhappy agents in any $G_i$. In $G_1$ only agent $a$, having cost $5$, is unhappy. There is only one vertex having distance $5$ to $a$, which is $l_1$. Thus, agent could swap towards any vertex having distance at most $3$ to $l_1$ to improve on this distance. Possible target-vertices are $b,c,e,g,j,k$ and $l_1$. But, after any such swap, agent $a$ must have distance $4$ to vertex $l_3$, which implies, that agent $a$ cannot decrease her cost by more than $1$. Furthermore, swapping towards $c,k$ or $l_1$ yields distance $5$ towards $l_3$, which rules out these target vertices. If agent $a$ performs the swap $ad$ to $ae$ in $G_1$, then we obtain network $G_2$. In $G_2$ only vertex $c$, having cost $4$, is unhappy. Her unique vertex in distance $4$ is $l_2$. Thus, a swap towards $a,d,e,j$ or $l_2$ would reduce this distance. However, only the swap towards $a$ is an improving move, since in all other cases agent $c$'s distance to $l_6$ increases to $4$. This swap 
transforms $G_2$ into $G_3$. In $G_3$ we have that agent $a$ and agent $b$ are unhappy. Agent $b$ in $G_3$ is in a 
similar situation as agent $a$ in $G_1$. Her only vertex in distance $5$ is the leaf $l_1$ and by swapping towards $a,c,e,g,j,k$ or $l_1$ this distance can be reduced. But any such swap yields that agent $b$'s distance to $l_3$ increases to at least $4$, which implies that a cost decrease by $1$ is optimal. Furthermore, the swaps towards $c,k$ or $l_1$ are not improving moves since these yield distance $5$ towards $l_3$. If agent $b$ performs the swap $bd$ to $be$ we obtain network $G_4$. Agent $a$ in $G_3$ has $l_3$ as her only vertex in distance $4$. There is exactly one swap for agent $a$, which decreases this distance to $3$ without increasing any other distance to more than $3$, and this is the swap $ae$ to $ad$. Note, that this swap of agent $a$ transforms $G_3$ to a network which is isomorphic to network $G_1$. Finally, we argue for agent $a$ in $G_4$. This agent is in a similar situation as agent $a$ in $G_3$. Her only vertex in distance $4$ is $l_3$ and the swap towards $d$ is the only swap which 
does not increase any other distance to more than $3$. Thus, this move is agent $a$'s unique best response in $G_4$ and 
this move leads to network $G_5$. 
\end{proof} 

If played on a non-complete host-graph, then we get the worst possible dynamic behavior. 
\begin{corollary}\label{cor_asg_host}
 The \textsc{Sum}-ASG and the \textsc{Max}-ASG on a non-complete host graph are not weakly acyclic.
\end{corollary}
\begin{proof}

 \textbf{\textsc{Sum}-version:} We use the best response cycle $G_1,\dots, G_4$ shown in Fig.~\ref{fig:asymswapbr} and let the host graph $H$ be the complete graph but without the edge $\{a,f\}$. In this case, agent $f$'s best response move in $G_1$ is her only possible improving move. For the networks $G_2, G_3$ and $G_4$ it is easy to check, that the respective moving player has exactly one possible improving move. 
 
 \textbf{\textsc{Max}-version:} We use the best response cycle $G_1,\dots,G_6$ from Fig.~\ref{fig:maxasgbr}. As host graph $H$ we use the complete graph, but without edges $\{a,b\}, \{a,g\}, \{a,j\}, \{b,g\}, \{b,j\}$. By inspecting the proof of Theorem~\ref{thm_maxasg_brcycle}, it is easy to see that in each step of the cycle the moving agent has exactly one improving move. 
 \end{proof}

\subsection{The Boundary between Convergence and Non-Convergence}\label{sec_asg_boundary}
In this section we explore the boundary between guaranteed convergence and cyclic behavior. Quite surprisingly, we can draw a sharp boundary by showing that the undesired cyclic behavior can already occur in $n$-vertex networks having exactly $n$ edges. Thus, one non-tree edge suffices to radically change the dynamic behavior of Asymmetric Swap Games. Our constructions are such that each agent owns exactly one edge, which corresponds to the uniform unit budget case, recently introduced by Ehsani et al.~\cite{Ehs11}. Hence, even if the networks are build by identical agents having a budget the cyclic behavior may arise. This answers the open problem raised by Ehsani et al.~\cite{Ehs11} in the negative.  

\begin{theorem}\label{thm_1edge_brcycle}
 The \textsc{Sum}-ASG and the \textsc{Max}-ASG admit best response cycles on a network where every agent owns exactly one edge. 
\end{theorem}
\begin{proofof}{Theorem~\ref{thm_1edge_brcycle}, \textsc{Sum}-version}
The network which induces a best response cycle and the steps of the cycle are shown in Fig.~\ref{fig:asym1edgecycle}. Let $n_k$ denote the number of vertices having the form $k_j$, for some index~$j$. 
\begin{figure}[!h]
  \centering
  \includegraphics[width=\textwidth]{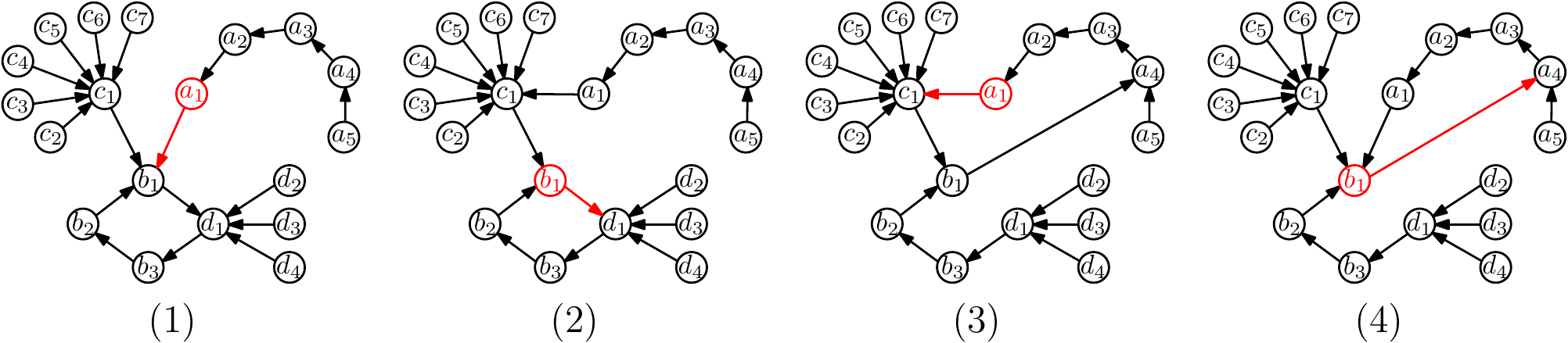}
  \caption{The steps of a best response cycle for the \textsc{Sum}-ASG where each agent owns exactly one edge.}
  \label{fig:asym1edgecycle}
\end{figure}

\noindent In the first step, depicted in Fig.~\ref{fig:asym1edgecycle}~(1), agent~$a_1$ has only one improving move, which is the swap from $b_1$ to $c_1$. This swap reduces agent~$a_1$'s cost by $1$, since $n_c = n_b + n_d +1$. 
After this move, shown in Fig.~\ref{fig:asym1edgecycle}~(2), agent~$b_1$ is no longer happy with her edge towards $d_1$, since by swapping towards $a_4$ she can decrease her cost by $2$. This is a best possible move for agent~$b_1$ (note, that a swap towards $a_3$ yields the same cost decrease). 
But now, in the network shown in Fig.~\ref{fig:asym1edgecycle}~(3), by swapping back towards vertex $b_1$, agent~$a_1$ can additionally decrease her distances to vertices $a_4$ and $a_5$ by $1$. This yields that agent~$a_1$'s swap from $c_1$ to $b_1$ decreases her cost by $1$. This is true, since all distances to $c_j$ vertices increase by $1$ but all distances to $b_i$ and $d_l$ vertices and to $a_4$ and $a_5$ decrease by $1$ and since we have $n_c = n_b + n_d + 1$. Note, that this swap is agent~$a_1$'s unique improving move.   
By construction, we have that after agent~$a_1$ has swapped back towards $b_1$, depicted in Fig.~\ref{fig:asym1edgecycle}~(4), agent~$b_1$'s edge towards $a_4$ only yields a distance decrease of $7$. Hence, by swapping back towards $d_1$, agent~$b_1$ decreases her cost by $1$, since her sum of distances to the $d_j$ vertices decreases by $8$. This swap is the unique improving move of agent~$b_1$ in this stage. Now the best response cycle starts over again, with agent~$a_1$ moving from $b_1$ to~$c_1$. 
\end{proofof}
\begin{proofof}{Theorem~\ref{thm_1edge_brcycle}, \textsc{Max}-version}
 Fig.~\ref{fig:maxasym1edgecycle} shows the steps of a best response cycle for the \textsc{Max} version in a network, where each agent owns exactly one edge. 
 \begin{figure}[!h]
  \centering
  \includegraphics[width=\textwidth]{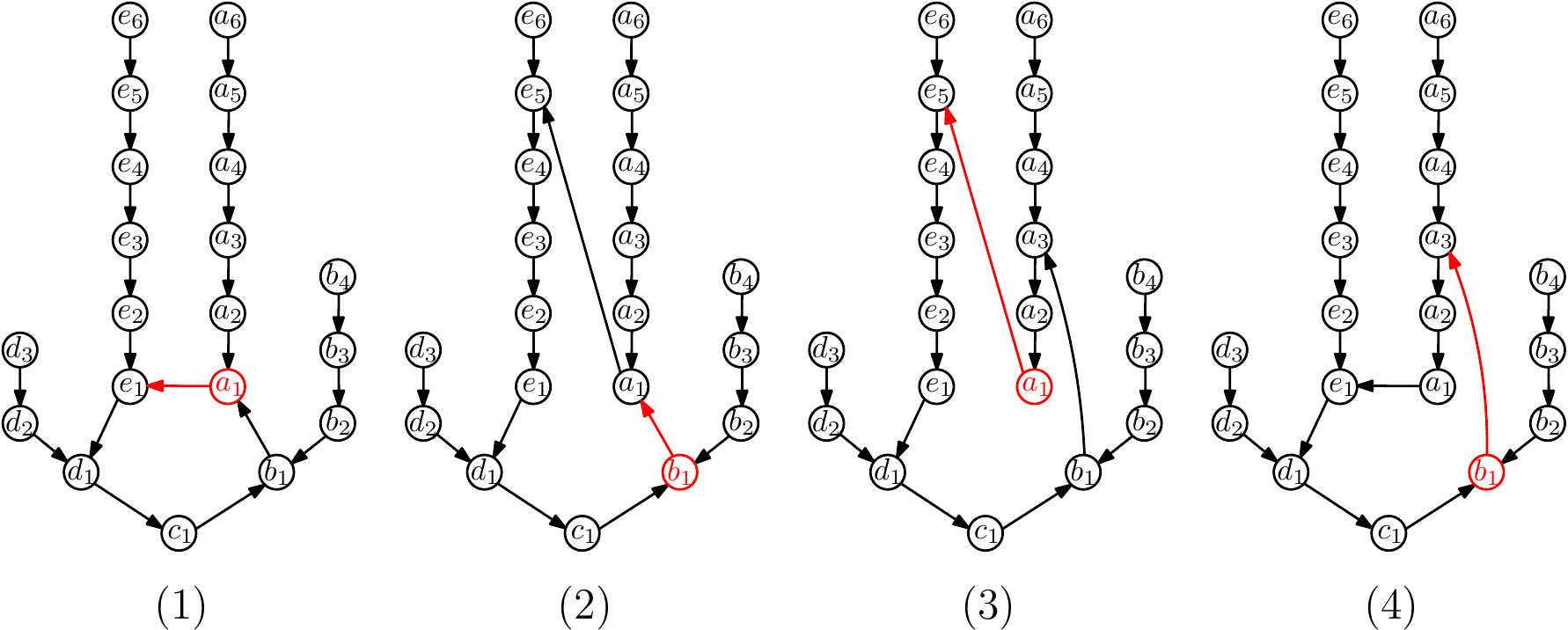}
  \caption{The steps of a best response cycle for the \textsc{Max}-ASG where each agent owns exactly one edge.}
  \label{fig:maxasym1edgecycle}
\end{figure}

In the first step of the cycle, shown in Fig.~\ref{fig:maxasym1edgecycle}~(1), agent $a_1$ can decrease her maximum distance from $6$ to $5$ by swapping from $e_1$ to one of the vertices $e_2,\dots,e_5$. Note, that all these swaps yield the same distance decrease of $1$ and, since $a_1$ has distance $5$ towards $a_6$ and swapping towards any of the $a_i$-vertices is obviously sub-optimal, no other swap can yield a larger cost decrease. By agent $a_1$ performing the swap towards $e_5$ we obtain the network in Fig.~\ref{fig:maxasym1edgecycle}~(2). 

Now, agent $b_1$ can improve her situation with a swap from $a_1$ to $a_2$ or to $a_3$. Both possible swaps reduce her maximum distance from $6$ to $5$. This is best possible: The cycle has length $9$ before agent $b_1$'s move, which implies that there are two vertices on the cycle which have distance $4$ to $b_1$. Observe, that agent $b_1$ must swap towards a vertex which has at most distance $4$ to vertex $a_6$ to reduce her maximum distance. Clearly, only one of the $a_i$ vertices, with $i\neq 1$, is possible. However, swapping away from $a_1$ must increase the cycle by at least $1$, which implies that after this swap agent $b_1$ must have at least one cycle-vertex in distance $5$. Hence, no swap can decrease agent $b_1$'s maximum distance by more than $1$. Let agent $b_1$ perform the swap towards $a_3$ and we end up with the network in Fig.~\ref{fig:maxasym1edgecycle}~(3).

Agent $a_1$ now finds herself sitting in a large cycle having maximum distance $7$ towards $d_3$ and distance $6$ to vertex $b_4$. By swapping from $e_5$ to one of the vertices $e_1,e_2,e_3$, agent $a_1$ can reduce her maximum distance to $6$. This is optimal, since all improving moves must swap towards a vertex having at most distance $5$ to vertex $d_3$ and agent $a_1$ cannot move to far away from vertex $e_6$. The vertices $e_1,e_2,e_3$ are the only vertices which satisfy both conditions. Let $a_1$ swap towards $e_1$ and we get the network depicted in Fig.~\ref{fig:maxasym1edgecycle}~(4).

In the last step of the best response cycle, we have agent $b_1$ with maximum distance $8$ towards vertex $e_6$. Clearly, agent $b_1$ wants to move closer to this vertex but this implies, that she must move away from vertex $a_6$. The only possible compromise between both distances is a swap either to $a_1$ or to $e_1$. Both these swaps yield a decrease of $b_1$'s maximum distance by $1$. By swapping towards vertex $a_1$, we end up with our starting configuration and the cycle is complete. 
\end{proofof}

\begin{remark}
 We can give best response cycles for both versions of the ASG for the case where every agent owns exactly two edges. We conjecture, that such cyclic instances also exist for all cases where every agent owns exactly $k$ edges, for any $k\geq 3$. In particular, it would be interesting if there is a generic construction which works for all $k\geq 1$. 
\end{remark}

\subsection{Empirical Study of the Bounded-Budget Version}\label{sec_asg_empirical}
We have conducted extensive simulations of the convergence behavior and the obtained results provide a sharp contrast to our mostly negative theoretical results for both versions of the ASG. Our experiments show for the bounded-budget version a surprisingly fast convergence in at most $5n$ steps under the max cost policy or by choosing the moving agents uniformly at random. Despite millions of trials we have not found any best response cycle in our experiments. This indicates that our negative results may be only very rare pathological examples. 

We first describe the experimental setup, then we will discuss the obtained results for the bounded-budget version of the ASG for the \textsc{Sum} and the \textsc{Max}-version of the distance-cost function. 
\subsubsection{Experimental Setup}\label{sec_asg_setup}
One run of our simulation can be described as follows: First we generate a random initial network, where every agent owns exactly $k$ edges. Then, for every step of the process, the move policy decides which agents is allowed to perform a best possible edge-swap. After the respective swap has happened, we again let the move policy decide which agent is allowed to move next. We count the number of steps until a stable network is found. Thus, the number of steps equals the number of performed moves.  

Computing a best possible edge-swap of an agent can be done in polynomial time by simply checking all possible edge-swaps and re-computing the cost. 

Under the max cost policy we calculate the agents' costs and check in descending order if the respective agent can perform an improving move. If we have found an unhappy agent, then we calculate the best possible edge-swap for this agent, breaking ties uniformly at random, and we let this agent perform the respective move. Then the process starts all over again until there is no unhappy agent left.   

Under the random policy we choose one agent uniformly at random and check if this agent can perform an improving move. If not, then we remove this agent from the set of candidates and we choose another agent uniformly at random from the remaining candidates. We proceed iteratively until we have found an unhappy agent or until no candidate is left. In the former case, we let this agent perform a best possible edge-swap and start all over again (with all agents being a possible candidate again). If the latter happens, then we stop. 

The initial network is generated as follows: We start with an empty graph $G$ on $n$ vertices. Then, to ensure connectedness, we create a random spanning tree among all $n$ agents as follows: We start with a uniformly chosen random pair of agents and insert the respective edge into $G$. The owner of this edge is chosen uniformly at random among its endpoints. We mark both vertices. Then, we iteratively choose one unmarked agent uniformly at random from the set of unmarked agents and one marked agent uniformly from the set of all marked agents and we insert the respective edge and mark the former agent. The edge-ownership is chosen uniformly at random with the constraint that no agent is allowed to own more than $k$ edges. If all agents are marked, then we have that $G$ is a spanning tree. Now we proceed inserting edges into $G$ as follows: First we mark all agents who already own $k$ edges. Then we iteratively choose one unmarked agent and one other agent uniformly at random and insert the edge with the 
first agent being its owner, if the edge is not already present in $G$. If the edge is present, then we randomly choose another suitable pair of agents. Again we mark agents having already $k$ edges. We stop, if there is no unmarked agent left.   
\subsubsection{Experimental Results and Discussion}
\paragraph{Results for the \textsc{Sum}-ASG:} 

Our obtained results for the \textsc{Sum}-ASG in the bounded-budget version can be found in Fig.~\ref{fig:exp_swap_sum}. We simulated 10000 runs for each configuration and the maximum of the observed convergence time for each configuration is plotted. Here a configuration consists of the number of agents in the initial network and the choice of the move policy.
\begin{figure}[t]
 \centering
  \includegraphics[width=0.495\textwidth]{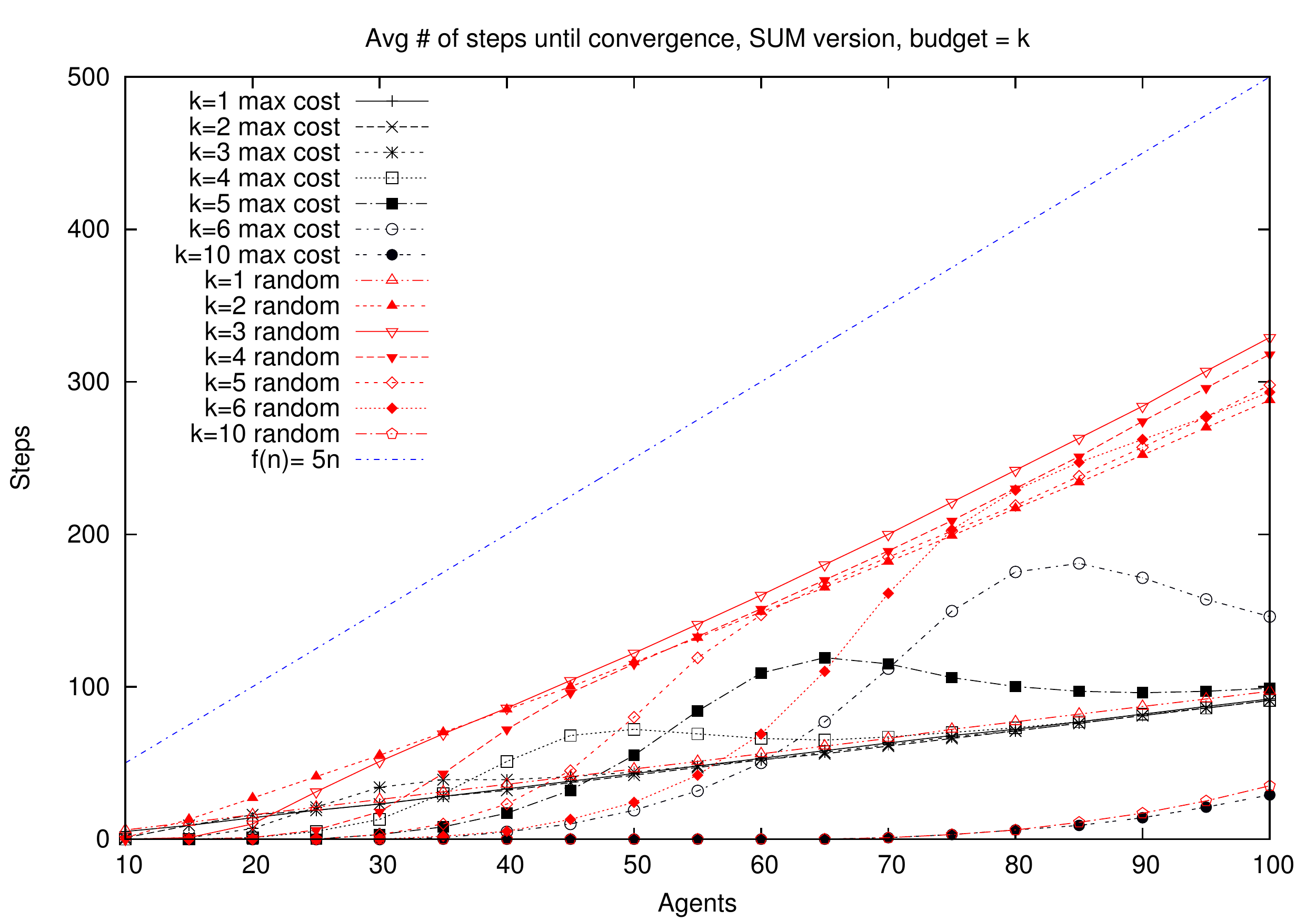}
  \includegraphics[width=0.495\textwidth]{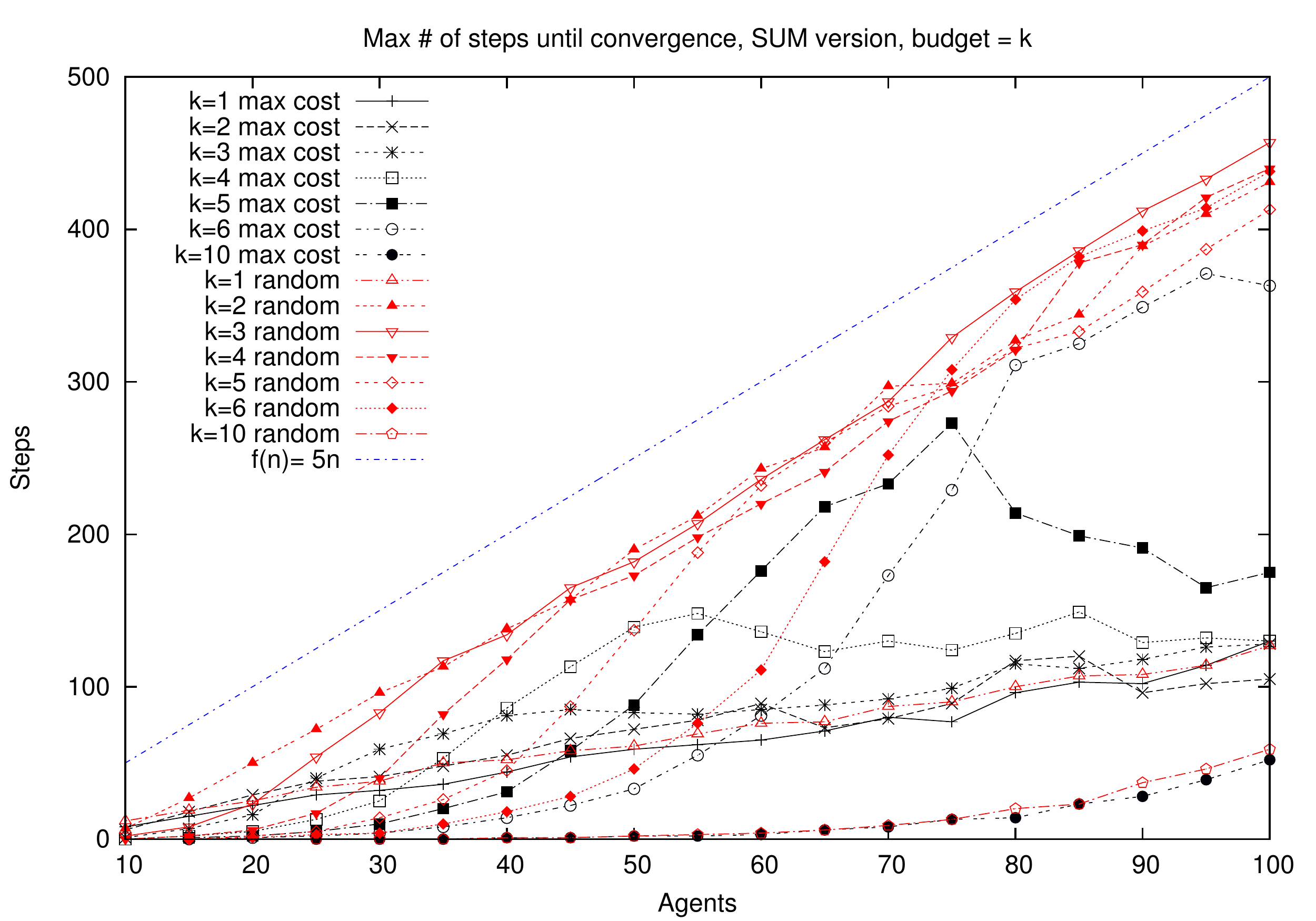}
  \caption{Experimental results for the \textsc{Sum}-ASG with budget $k$. The average number of steps needed for convergence is plotted on the left, the maximum number of steps needed for convergence is plotted on the right. Each point is the average/maximum over the number of steps needed for convergence of 10000 trials with random initial networks where each agent owns exactly $k$ edges.}
    \label{fig:exp_swap_sum}
\end{figure}
The results for the max cost policy are plotted in black, whereas the results for the random policy are shown in red. 

First of all, note that no run took longer than $5n$ steps, where $n$ is the number of agents, and that the max-cost policy yields faster convergence than the random policy. The only exceptions here are the random policy in the cases where $k=1$ and $k=10$, which shows roughly the same behavior than the max cost policy. For the case $k=10$ the number of agents seems to small to produce the difference. Indeed, for all other budgets we see that the plots for the random policy and the max cost policy are close together for small numbers of agents and start to separate as $n$ grows larger. Note, that for $k=1$ only roughly $n$ steps are needed for convergence under both move policies. This is to be expected for the max cost policy since the initial network is almost a tree and we have shown in Corollary~\ref{cor_mcbrd} that the \textsc{Sum}-ASG on trees converges in at most $n+\lceil n/2 \rceil -5$ steps.  

For $k>1$ the results are particularly interesting. Under the max cost policy our simulations reveal a rather curious behavior: The convergence time increases super linear until it peaks and then, for larger $n$, it converges to $n$. One possible explanation for this is the ratio of edges versus non-edges in the network. For small $n$ we have that the initial networks are very dense, which yields that agents have very short distances to each other. This implies, that moves yield a relatively low cost-decrease since only a few distances can be reduced, which implies that after a small number of steps, no such move is available and the convergence process stops. For large $n$ we have that the initial networks are very sparse, which implies that agents ``at the perimeter'' can achieve a large cost-decrease by performing a move. If only such agents move, then we have a sequence of moves which reduce a high number of individual distances and this leads to fast convergence. The slowest convergence time is achieved 
if the ratio of present edges over all possible edges is between $\tfrac{1}{7}$ and $\tfrac{1}{6}$. 

Under the random policy, we see a completely different behavior if $k>1$. For example, note the difference between the $k=2$ case of the random versus the max cost policy.  Here we have the expected behavior that the convergence time is strictly increasing for larger $n$. Interestingly, except for small $n$, the convergence time grows only linear in $n$. The super linear increase for small $n$ can be explained as under the max cost policy. The initial networks are too dense to admit a high number of best possible improving moves. For large $n$, the results can be explained as follows: in contrast to the max cost policy we have that the random policy often picks agents who already have a relatively central position in the network. Such agents can improve only slightly. Thus, moves of such agents only reduce a small number of individual distances, which explains why a lot of agents remain unhappy with their situation.

\paragraph{Results for the \textsc{Max}-ASG:} 

The results for the \textsc{Max}-ASG under both move policies can be found in Fig.~\ref{fig:exp_swap_max}. As in the \textsc{Sum}-version, we simulated 10000 runs for each configuration. The results for the max cost policy are shown in black whereas the results for the random policy are plotted in red.
\begin{figure}[!h]
 \centering
  \includegraphics[width=0.495\textwidth]{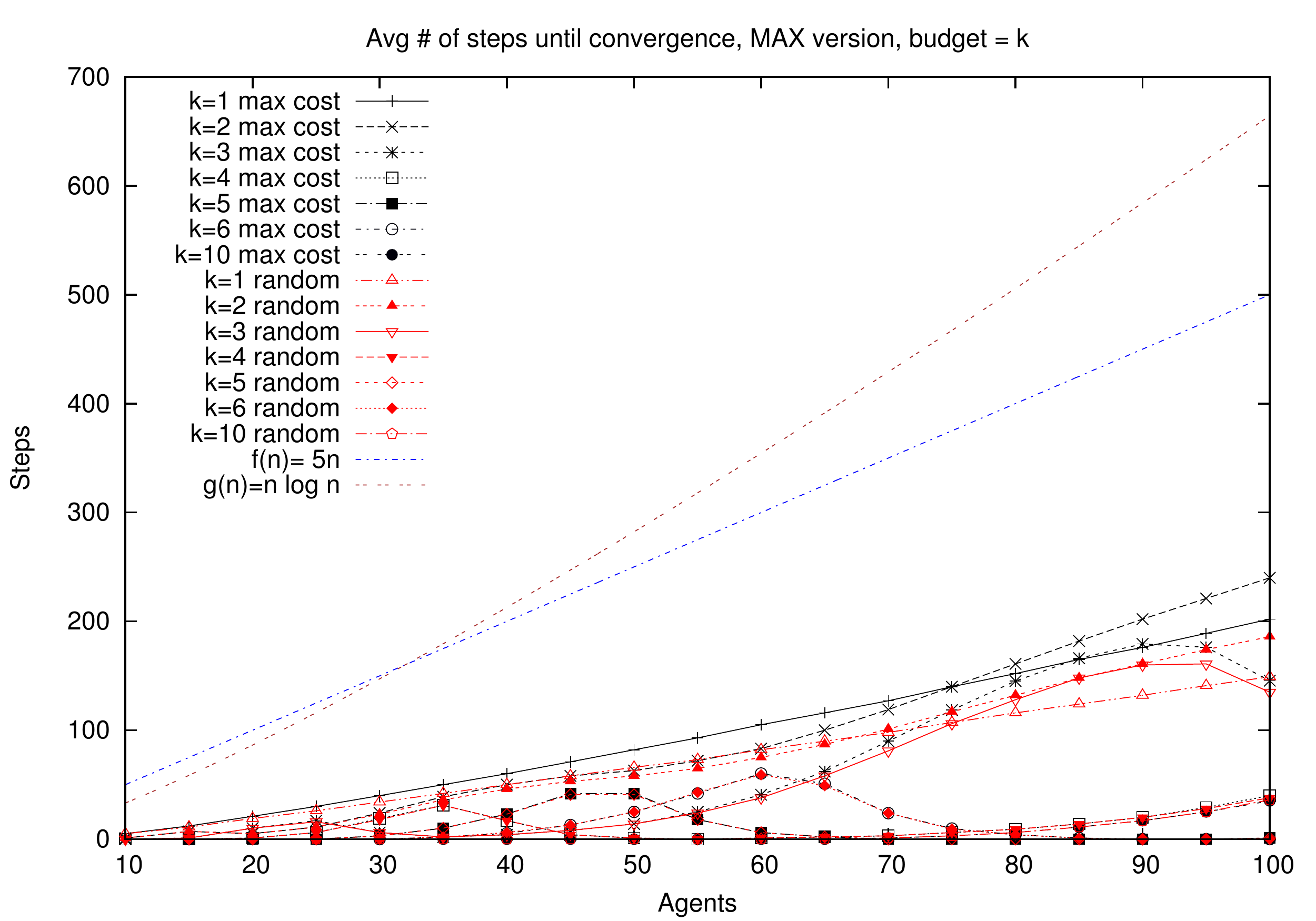} 
  \includegraphics[width=0.495\textwidth]{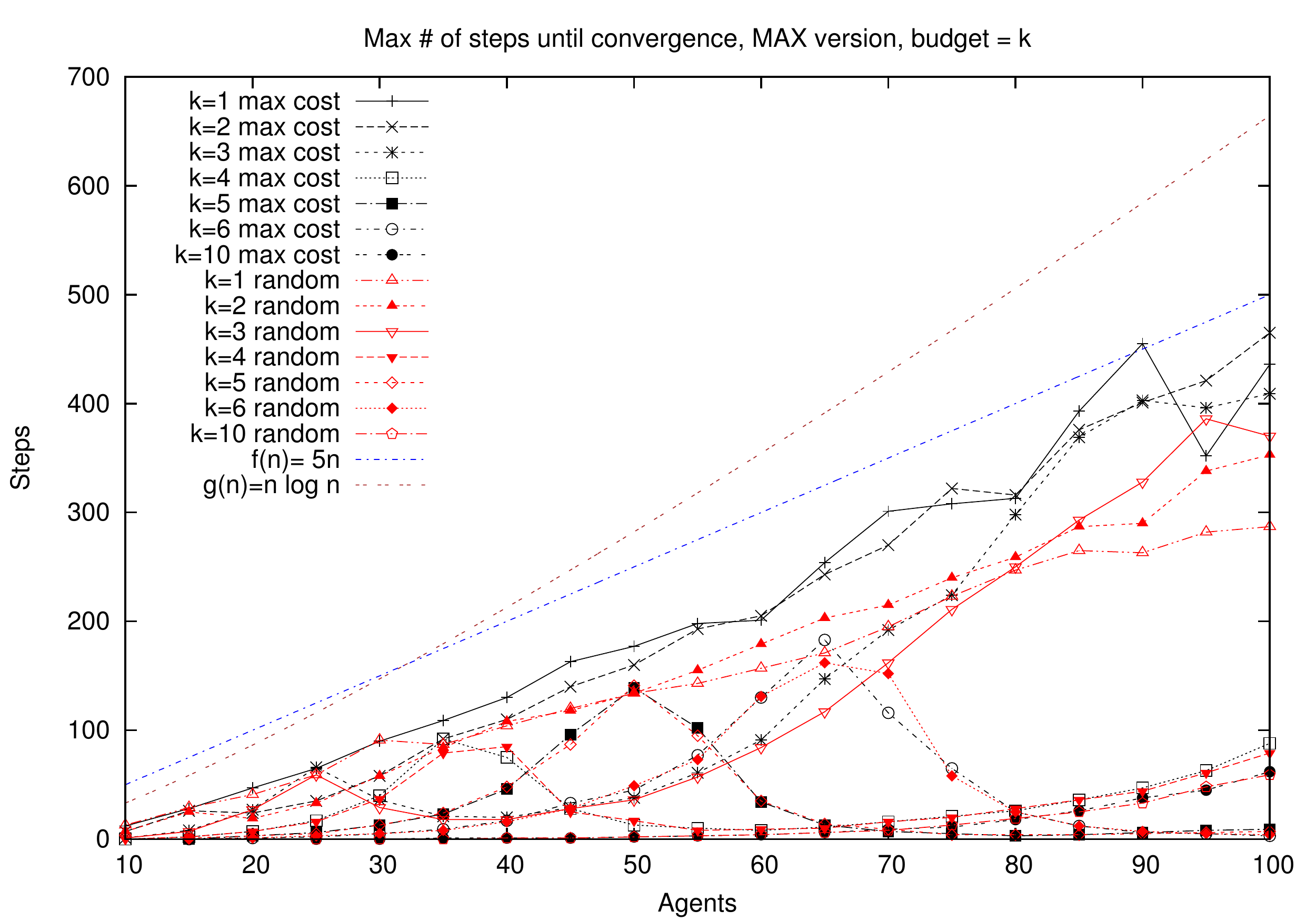}
  \caption{Experimental results for the \textsc{Max}-ASG with budget $k$. The average number of steps needed for convergence is plotted on the left, the maximum number of steps needed for convergence is plotted on the right. Each point is the average/maximum over the number of steps needed for convergence of 10000 trials with random initial networks where each agent owns exactly $k$ edges.}
    \label{fig:exp_swap_max}
\end{figure}

The plots show that, with one exception, every run of our simulations converged in less than $5n$ steps. Hence, the \textsc{Max}-version yields the same fast convergence to a stable network as in the \textsc{Sum}-version. However, there is a crucial difference: Here the different move policies do not yield a different convergence behavior. The curves for the random policy and the max cost policy for the same budget $k$ are very close to each other. For $k\geq 4$ they are almost indistinguishable. For smaller $k$ we see that the random policy slightly outperforms the max cost policy. The reason for the similar behavior under both move policies is that in the \textsc{Max}-version there are significantly more agents which all have maximum cost than in the \textsc{Sum}-version. After a small number of steps we even have the situation that a large fraction of all agents has maximum cost. Thus, for an agent we 
have that choosing randomly among the maximum cost agents and choosing randomly among all agents yields no significant difference in the probability of being chosen.  

Observe, that for $k=1$ the convergence time of the max cost policy grows slightly super linear and is well below $n\log n$. In this case the network is almost a tree and this explains why the convergence time is very close to our bound of $\Theta(n\log n)$ shown in Corollary~\ref{cor_mcbrd}. Interestingly, we have for both move policies that the convergence time decreases as the budget $k$ increases. Intuitively this is not surprising, since having more edges leads to a lower distance-cost for an agent and to a lower diameter in the network. However, agents may have many other agents in maximum distance and the respective shortest paths may not overlap. It is not clear how a higher budget can help in such a situation. 

The curves for $k\geq 3$ have a rather strange shape. They show a local peak right after the beginning and then they fall and much later rise again. We cannot explain this curious behavior but we conjecture that the reason is the density and the distance-distribution in the initial networks.
\paragraph{Further Remarks:}
We emphasize again, that among all these simulations we have never encountered a cyclic instance. Our cyclic constructions basically rely on the observation that it can happen that an improving move of one agent increases the costs of several agents simultaneously. Empirically this happens very often but has no severe consequences since whenever the other agents are allowed to move, they can compensate for their temporary cost increase. As we have shown, it requires a rather intricate sequence of such moves to achieve that the costs of all agents after several steps are the same as before. Such sequences are therefore very unlikely to occur.       

\section{Dynamics in (Greedy) Buy Games}
We focus on the dynamic behavior of the Buy Game and the Greedy Buy Game. Remember, that we assume, that each edge can be created for the cost of $\alpha>0$. 
\subsection{Convergence Results}\label{sec_bg}
We show that best response cycles exist, even if arbitrary strategy-changes are allowed. However, on the positive side, we were not able to construct best response cycles where only one agent is unhappy in every step. Hence, the right move policy may have a substantial impact in (Greedy) Buy Games. In contrast to this, we rule out this glimmer of hope if played on a non-complete host-graph.
\begin{theorem}\label{thm_buy}
 The \textsc{Sum}-(G)BG and the \textsc{Max}-(G)BG admit best response cycles. 
\end{theorem}
\begin{proofof}{Theorem~\ref{thm_buy}, \textsc{Sum}-version}
 We prove both statements by giving a best response cycle, where the best response of any moving agent consists of either buying, deleting or swapping one edge. The best response cycle $G_1,\dots,G_6$, for $7 < \alpha <8$, is depicted in Fig.~\ref{fig:nashbr}.
 \begin{figure}[!h]
  \centering
  \includegraphics[width=15cm]{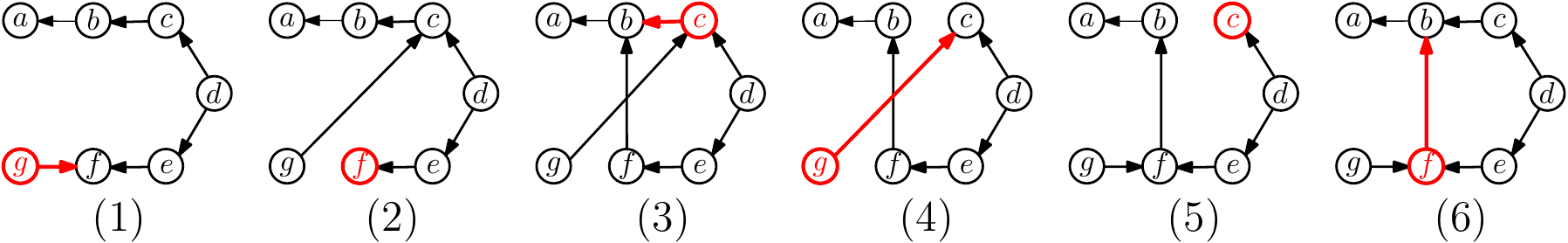}
  \caption{The steps of a best response cycle for the \textsc{Sum}-(G)BG for $7 < \alpha < 8$.}
  \label{fig:nashbr}
\end{figure}

\noindent We analyze the steps of the cycle and show, that the indicated strategy-change is indeed a best response move -- even if there are no restrictions on the admissible strategies. 

In network $G_1$, it is obvious, that agent $g$ is unhappy with her situation. The indicated swap $gf$ to $gc$ decreases agent $g$'s cost from $\alpha + 21$ to $\alpha + 15$. This is a best possible move, which can be seen as follows. Clearly, deleting her unique own edge would disconnect the network. Hence in all optimal strategies agent $g$ must purchase at least one edge. Among all strategies, where $g$ buys exactly one edge, that is, among all possible single edge-swaps, we have that buying an edge towards a vertex having minimum cost in $G_1 - g$ is optimal. Here, $G_1 - g$ is the network $G_1$ with vertex $g$ removed. Thus, swapping her edge towards vertex $c$ is optimal. Buying exactly $1<k\leq6$ edges yields cost of at least $k\alpha + k + 2(6-k) > 6k+12 \geq 24$, which is no improvement since $\alpha + 15 < 23$. After agent $g$ has performed her strategy-change we obtain network $G_2$.

In $G_2$ we claim that agent $f$ is unhappy and that her best possible move is to buy an edge towards vertex $b$. First of all, this is an improving move, since the edge $fb$ decreases her cost from $19$ to $11+\alpha$, which is a strict cost decrease since $\alpha < 8$. The target vertex $b$ is optimal, since connecting to $c$ yields the same cost and connecting to any other vertex yields a higher cost. Clearly, agent $f$ cannot delete or swap any edges. Furthermore, buying at least two edges yields cost of more than $19$, since $2\alpha > 14$ and there are six other vertices in $G_2$ to which $f$ must have distance of at least $1$. The edge purchase of agent $f$ leads us to network~$G_3$.

In network $G_3$ we claim that agent $c$ is unhappy and that her best possible move is to delete her edge towards $b$. Agent $c$ has cost $9+\alpha$ in $G_3$. Deleting edge $cb$ yields cost $16 < 9+\alpha$, since $\alpha > 7$. Clearly, no strategy which buys at least two edges can be optimal for agent $c$, since $6 + 2\alpha > 16$. On the other hand, swapping her unique edge away from $b$ must increase agent $c$'s cost since at least one distance increases to $3$. If agent $c$ deletes her edge $cb$, then we obtain network~$G_4$.

In $G_4$, we have that agent $g$ is in a similar situation as she was in $G_1$. Agent $g$ is again a leaf-vertex of a path of length $6$. Thus, by an analogous argument as for agent~$g$ in $G_1$, we have that the swap $gc$ to $gf$ is a best possible move for agent $g$ in $G_4$. This move leads us to network~$G_5$.

In network $G_5$ we have that agent $c$ is in a similar situation as agent $f$ in $G_2$. By analogous arguments, it follows that buying the edge towards $b$ is a best possible move of agent $c$ in $G_5$. This edge-purchase transforms $G_5$ into $G_6$.

Finally, in network $G_6$ we have that agent $f$ is in a similar situation as agent $c$ in $G_3$. Thus, by analogous arguments, we have that deleting her edge $fb$ is an optimal move for agent $f$ in $G_6$. This deleting transforms $G_6$ into $G_1$ and we have completed the cycle. 
\end{proofof}
\begin{proofof}{Theorem~\ref{thm_buy}, \textsc{Max}-version}
 We give a best response cycle, where in every step of the cycle the moving agent has a best response which consists of a greedy move, that is, the strategy-change is the addition, deletion or swap of one edge. The best response cycle $G_1,\dots,G_4$, for $1<\alpha < 2$, can be seen in Fig.~\ref{fig:maxnashbr}.
 \begin{figure}[!h]
  \centering
  \includegraphics[width=13cm]{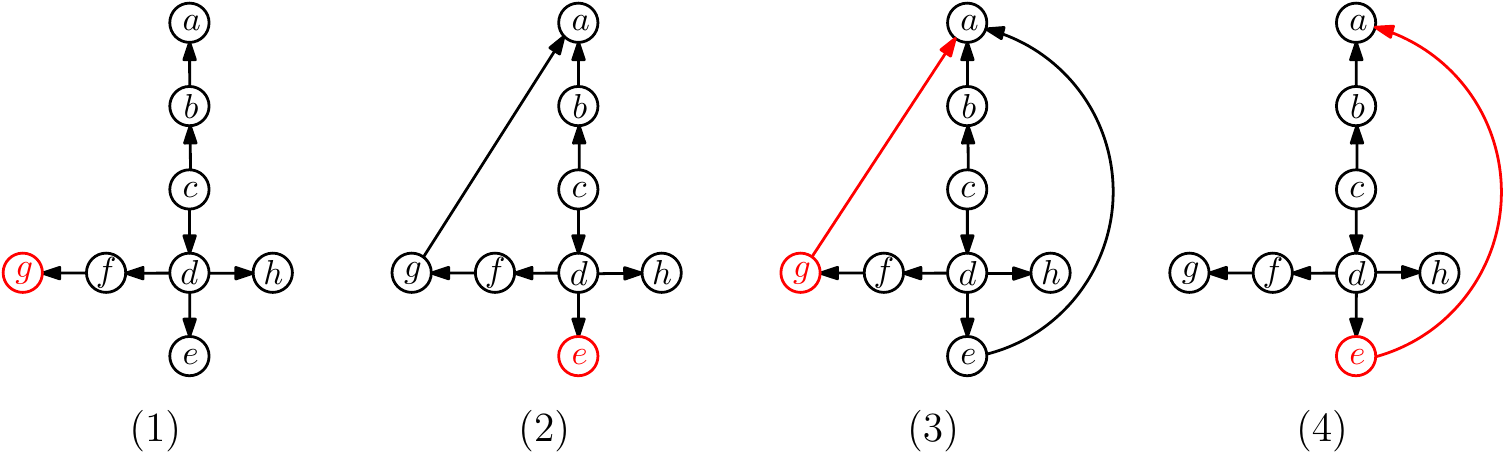}
  \caption{The steps of a best response cycle for the \textsc{Max}-(G)BG for $1 < \alpha < 2$.}
  \label{fig:maxnashbr}
\end{figure}

We show for each step of the cycle, that the indicated moving agent indeed performs a best response move which transforms the network of one step into the network of the next step of the cycle. 

In network $G_1$ we claim that agent $g$ is unhappy and that a best response move is to buy the edge $ga$. Clearly, agent $g$ cannot delete or swap any edges. Hence, it suffices to analyze all buy or multi-buy operations. Agent $g$ has cost $5$ in $G_1$. The purchase of edge $ga$ is an improving move, since with this yields a distance-cost of $3$ for agent $g$ and since $\alpha < 2$. Note, that agent $g$ must buy at least one edge to reduce her distance-cost in $G_1$. Furthermore, it is easy to see that with one additional edge a distance-cost of $3$ is best possible. Now, observe, that if agent $g$ buys more than one edge, than her distance-cost may decrease, but it cannot decrease by more than $1$ per edge. Since $\alpha > 1$, no strategy which buys at least two edges can yield strictly less cost than $3+\alpha$. The indicated move of agent $g$ transforms $G_1$ into $G_2$.

In $G_2$, agent $e$, having cost $4$, is unhappy with her situation. Buy buying the edge $ea$, agent $e$ can decrease her distance-cost to $2$. Since $\alpha < 2$, it follows that this is a improving move. Note, that agent $e$ cannot delete or swap any edges and that distance-cost of $2$ is optimal, unless agent $e$ buys $5$ edges, which clearly is too expensive. Hence, buying the edge $ea$ is a best response move for agent $e$ and this move leads to network $G_3$. 

In network $G_3$, we have that agent $g$, having cost $3+\alpha$, is unhappy. By deleting her own edge $ga$, agent $g$ can achieve a cost of $4$, which is strictly less than $3+\alpha$, since $\alpha > 1$. We claim that deleting edge $ga$ is a best response move for agent $g$. If agent $g$ swaps her unique own edge, than she cannot achieve a distance-cost of less than $3$. Thus, no swap can be an improving move. If agent $g$ buys at least one additional edge, then such a purchase may decrease agent $g$'s distance-cost by $1$ per edge, but since $\alpha >1$, this cannot outperform her current strategy in $G_3$. Thus, deleting edge $ga$ is the only improving move of agent $g$ and, thus, must be her best response move. This move transforms network $G_3$ in to $G_4$.

Finally, in network $G_4$ we have that agent $e$, having cost $3+\alpha$, is unhappy. Deleting her edge $ea$ yields a cost of $4$, which is strictly less than $3+\alpha$. Swapping this edge cannot decrease her distance-cost below $3$, which rules out any edge-swaps. If agent $e$ buys at least one additional edge, then she can reduce her distance-cost by at most $1$ per additional edge. Clearly, no such strategy may yield less cost than $3+\alpha$ and, thus, cannot be an improving move. Hence, deleting her edge $ea$ is her unique improving move, which must be her best response move. This move transforms $G_4$ into $G_1$. 
\end{proofof}

If we restrict the set of edges which can be build, then we get the worst possible dynamic behavior. In this case there is no hope for convergence if agents are only willing to perform improving moves. 
\begin{corollary}\label{cor_bg_host}
The \textsc{Sum}-(G)BG and the \textsc{Max}-(G)BG on general host graphs is not weakly acyclic. 
\end{corollary}
\begin{proof}

 \textbf{\textsc{Sum}-version:} We use the best response cycle $G_1,\dots,G_6$, shown in Fig.~\ref{fig:nashbr} with $7<\alpha <8$ and set the host-graph $H$ to the graph $G_1$ augmented by the two additional edges $bf$ and $cg$. With this host graph, we have that in every step of the cycle exactly one agent is unhappy and this agent has exactly one possible improving move. It follows, that starting with network $G_1$ on host-graph $H$ there is no sequence of improving moves which leads to a stable network.  
 
 \textbf{\textsc{Max}-version:} We use the best response cycle $G_1,\dots,G_4$ in Fig.~\ref{fig:maxnashbr} with $1<\alpha <2$ and we let the host-graph $H$ be the graph $G_1$ with the additional edges $ag$ and $ae$. It follows that in every step of the cycle, exactly one agent is unhappy and this agent has exactly one improving move. 
\end{proof}

\subsection{Empirical Study of Greedy Buy Games}\label{sec_gbg_empirical}
We give empirical results for the convergence time for both versions of the GBG. Note, that a best response for both versions of the GBG can be computed in polynomial time, whereas this problem is well-known~\cite{Fab03,MS10} to be NP-hard for the BG. Our results show a remarkably small number of steps needed for convergence in these games, which indicates that distributed local search is a practical method for selfishly creating stable networks. For the \textsc{Sum}-GBG no run took longer than $7n$ steps to converge, whereas for the \textsc{Max}-version we always observed less than $8n$ steps until convergence. Moreover, as in the simulations for the ASG, despite several millions of trials we did not encounter a cyclic instance. This indicates that such instances are rather pathological and may never show up in practice. 

As for our experiments in the ASG, we will first describe the experimental setup and then we summarize the obtained results for the convergence time of the \textsc{Sum}-GBG and the \textsc{Max}-GBG.
\subsubsection{Experimental Setup}\label{sec_exp_setup}
Our simulations for the GBG have a similar setup as the experiments presented in Section~\ref{sec_asg_setup}. One run of our simulations consists of the generation of a random initial network, then we employ the max cost or the random policy until the process converges to a stable network. See Section~\ref{sec_asg_setup} for a description of these move policies. We measure the number of steps needed for this convergence to happen.

We focus on the GBG, since it is well-known~\cite{Fab03} that in the BG computing a best response for an agent is NP-hard. In contrast, for the GBG this can be done in polynomial time by simply checking the best possible edge-deletion, edge-swap and edge-addition and re-computing the incurred cost, see~\cite{L12}. If at least two of these operations yield the same cost decrease, then we prefer deletions before swaps before additions. Such ties are very rare and we have obtained similar results by changing this preference order. 

The initial networks are generated analogously to the ASG, but this time we do not have to enforce the budget-constraint. Starting from an empty graph on $n$ vertices we first generate a random spanning tree to enforce connectedness of our networks. Then we insert edges uniformly at random until the desired number of edges is present. Note that we do not allow multi-edges. The ownership of every edge is chosen uniformly at random among the endpoints. In order to investigate the impact of the density of the initial network on the convergence time, we fix the number of edges in the initial network to be $n$, $2n$ and $4n$, respectively. The impact of the edge-cost parameter $\alpha$ is investigated by setting $\alpha$ to $n/10$, $n/4$, $n/2$ and $n$, respectively. Demaine et al.~\cite{De07} argue that this is the most interesting range for $\alpha$, since implies that the average distance is roughly on par with the creation cost of an edge.
\subsubsection{Experimental Results and Discussion}
We have simulated both the \textsc{Sum}-GBG and the \textsc{Max}-GBG. For each configuration we computed $5000$ runs. Here a configuration is determined by the number of agents, the number of edges $m$ in the initial network, the choice for $\alpha$ and the choice for the move policy. For sake of presentation, our plots only contain the results for $m=n$ and $m=4n$ and $\alpha=n$, $\alpha=n/4$ and $\alpha=n/10$.
\paragraph{Results for the \textsc{Sum}-GBG:}
Our results for the \textsc{Sum}-GBG can be found in Fig.~\ref{fig:exp_greedy_sum}. It can be seen that the convergence time grows roughly linear in $n$ for all configurations, which implies that these processes scale very well.  
\begin{figure}[bh]
 \centering
  \includegraphics[width=0.495\textwidth]{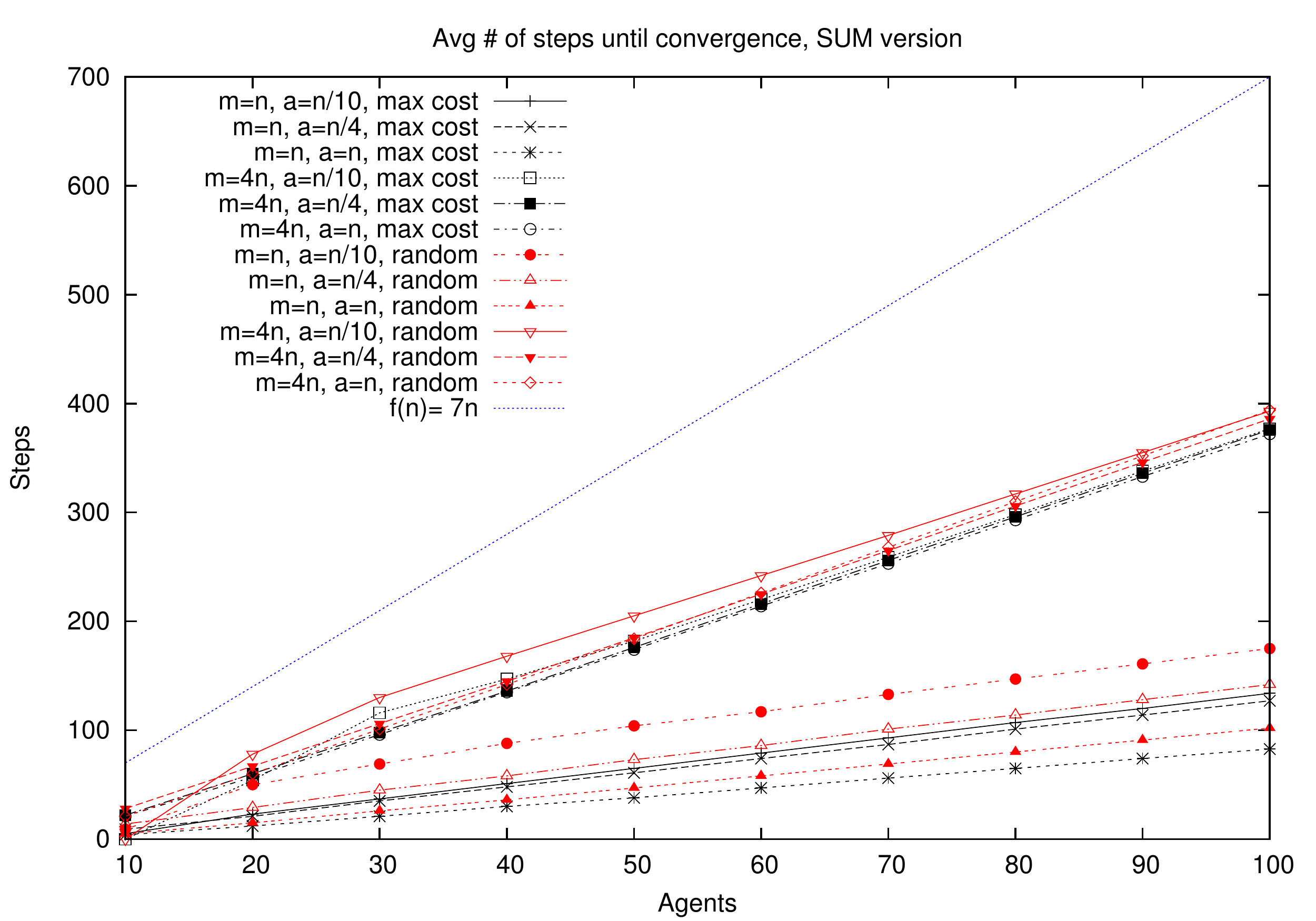}
  \includegraphics[width=0.495\textwidth]{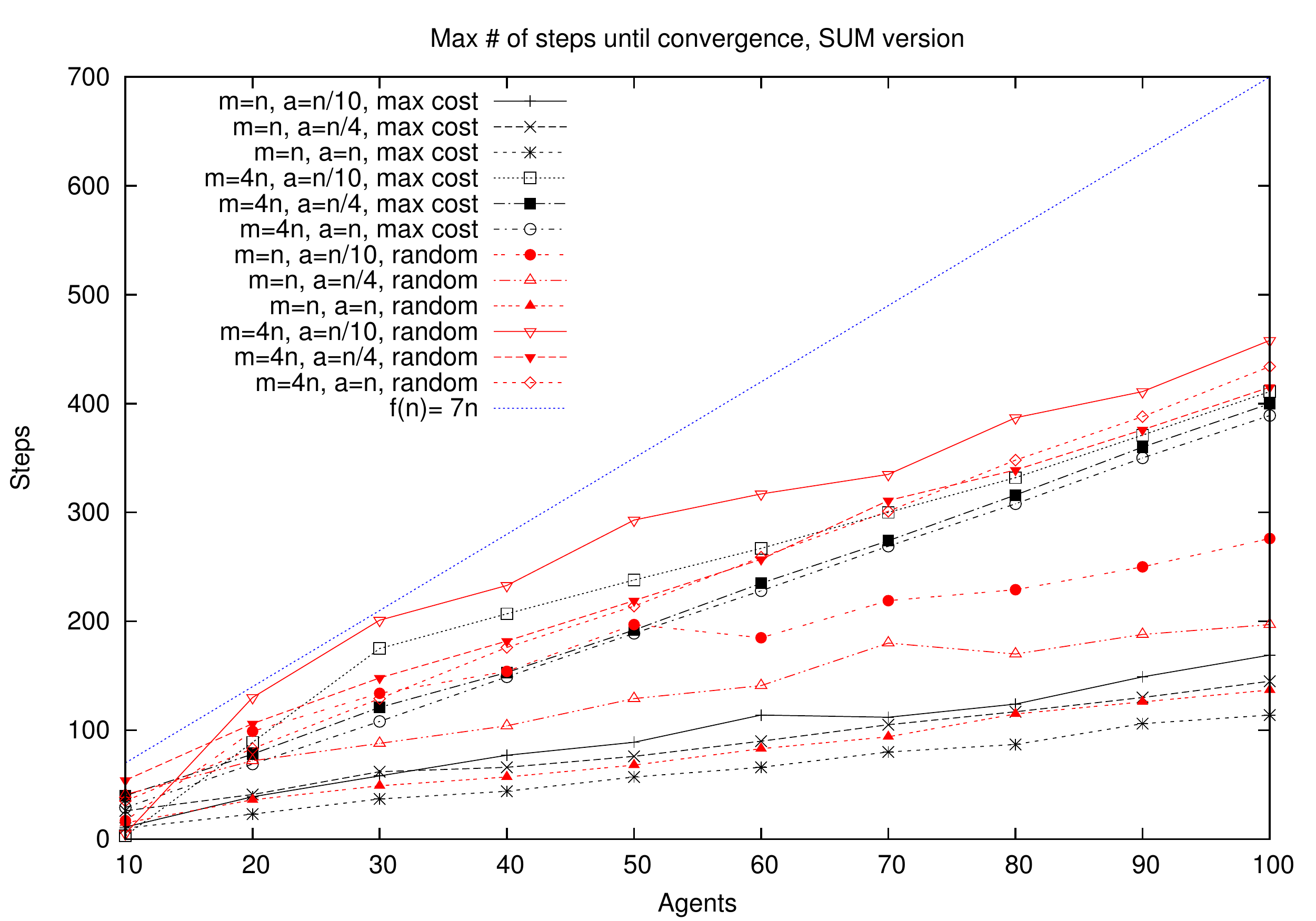}
  \caption{Experimental results for the \textsc{Sum}-GBG. The average number of steps needed for convergence is plotted on the left, the maximum number of steps needed for convergence is plotted on the right. Each point is the average/maximum over the number of steps needed for convergence of 5000 trials with random initial networks having $m$ edges and $\alpha = a$.}
    \label{fig:exp_greedy_sum}
\end{figure}
 
Similarly as for the ASG in the \textsc{Sum}-version, we have that the max cost policy outperforms the random policy. However, since the respective curves look similar, there seems to be no qualitative difference between both policies. For the cases where the initial network has $4n$ edges we have that the gap between the max cost policy and the random policy is even smaller. The cases where $m=2n$, which are not shown here, fit in well with this observation. This can be explained as in the ASG: The max cost policy favors agents ``at the perimeter'' whose best response moves decrease a high number of individual distances if these moves turn out to be swaps or additions. In contrast, the random policy picks agents which occupy a more central position in the network with higher probability. If such agents happen to swap or add an edge, then this decreases a smaller number of individual distances. This also explains why the gap becomes smaller, when the number of edges in the initial network becomes larger. 
With 
more edges to start with, it follows that the max cost agents have a more central position.     

Interestingly, the number of edges in the initial network seems to have an impact on the convergence time, since all curves for $m=4n$ are well above the respective curves for $m=n$. In the cases where $m=2n$, we have that the respective curves lie between the shown curves. The reason for this may be the relatively high value for $\alpha$ compared to the diameter of the resulting networks. We have not found any stable network having a diameter larger than $4$, which implies for our values of $\alpha$ that almost all stable networks happened to be stars. Since stars have $n-1$ edges, clearly $m-(n-1)$ deletion-steps happen during the convergence process. However, note that the convergence time is well above $m-(n-1)$, which shows that swap and add operations also occur fairly often. A typical sample trajectory for the case of $m=4n$ with $\alpha = n/4$ under the random policy looks as follows: First there is a phase with mostly deletions. Then this phase is followed by a phase where mostly swap and some buy 
and deletion operations occur. This is 
followed by a final phase where some swaps and mostly deletions happen. The first phase seems to be due to the fact that by our generation process the number of owned edges in the initial network varies more or less strongly between all agents and that $\alpha$ is relatively high. Agents having too many edges first try to get rid of them until their cost is in a range, where swaps or additions can outperform deletions, that is, where the distance-cost starts to dominate the edge-cost of that agent. The latter seems to explain what happens in the second phase. Here agents are mainly trying to minimize their distance-cost. After that, a high number of individual distances has decreased, which explains why in the last phase some swaps and a lot of deletions occur. Clearly, with less edges to start with, the second phase starts much earlier and the last phase is shorter as well. Under the max cost policy we observed that the first phase and the second phase are slightly shorter than under the random policy. In 
the first phase almost all operations are deletions, whereas under the random policy we see more swaps in the first phase. Interestingly the second phase under the max cost policy consists almost entirely of swap operations. Under the random policy we see more deletions and more add-operations in the second phase. Curiously, we generally see more add-operations under the random policy than under the max cost policy. This is counter-intuitive, since at least in the second phase agents with high cost have high distance-cost and therefore add-operations should be appealing for them. However, these small observations can serve as an additional explanation why the max cost policy outperforms the random policy.

The choice of $\alpha$ also has an influence on the convergence time. We see that a smaller $\alpha$ generally yields a higher number of steps needed for convergence. Again, our additional results for $\alpha=n/2$ confirm this observation. With smaller $\alpha$, we have that the length of the first phase decreases and the length of the second phase increases. It seems that agents buy more edges in the second phase which then leads to a longer last phase which may explain the overall increase in convergence time.

To clarify the influence of different starting topologies, we have compared three different types of initial networks. The results can be found in Fig.~\ref{fig:exp_compare_sum}. In the \texttt{random} setting, we focus on the initial networks with $n$ vertices and $n$ edges as described in Section~\ref{sec_exp_setup}. In the random line, or \texttt{rl}, setting we generate the initial networks as follows: first a path having $n$ vertices is created and then we choose the ownership of each edge uniformly at random among its endpoints. In the directed line, or \texttt{dl}, setting we generate a path having $n$ vertices and then the edge ownership is chosen such that all edges point in the same direction, that is, that the edge-ownership forms a directed path.    
\begin{figure}[!h]
 \centering
  \includegraphics[width=0.495\textwidth]{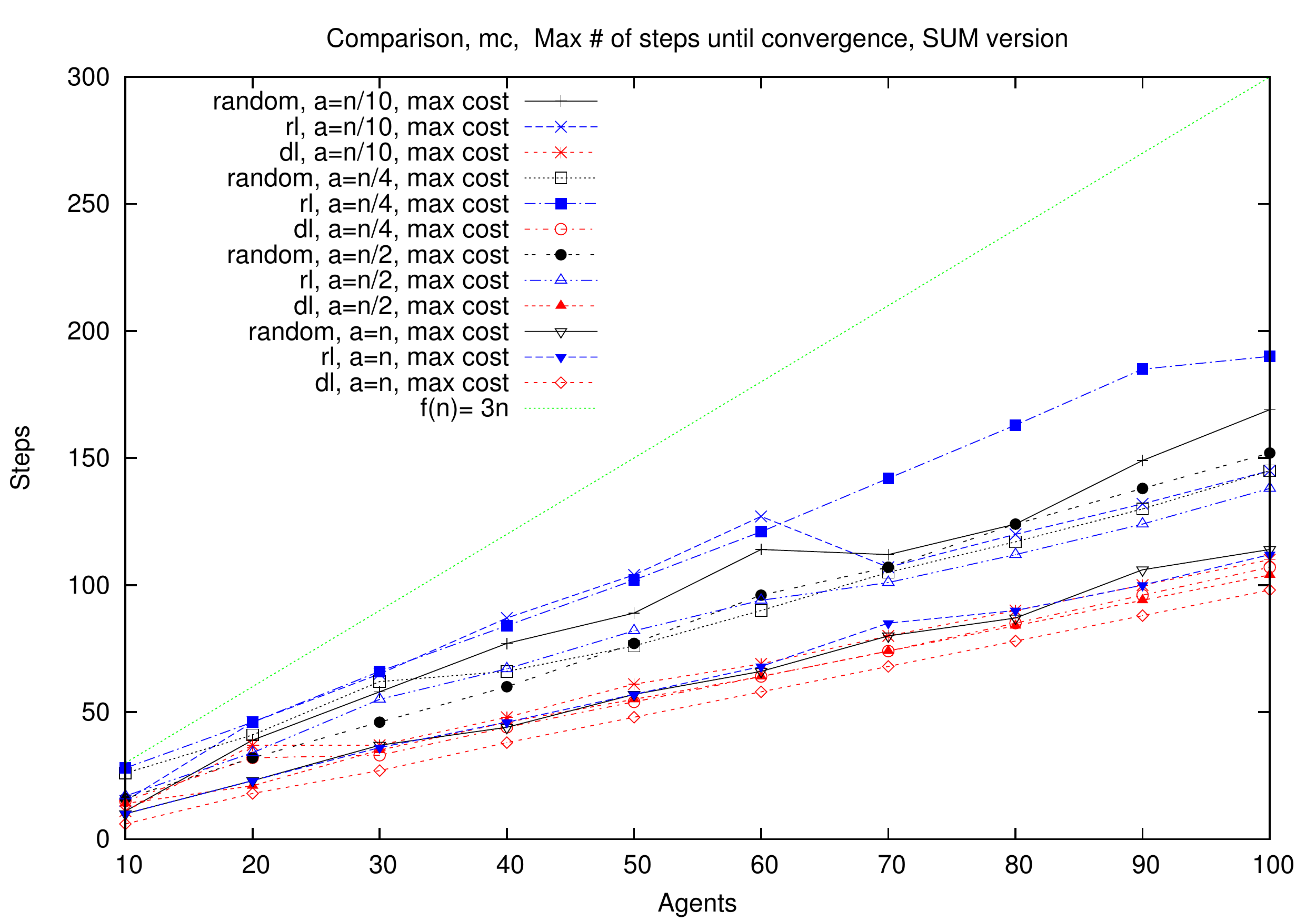}
  \includegraphics[width=0.495\textwidth]{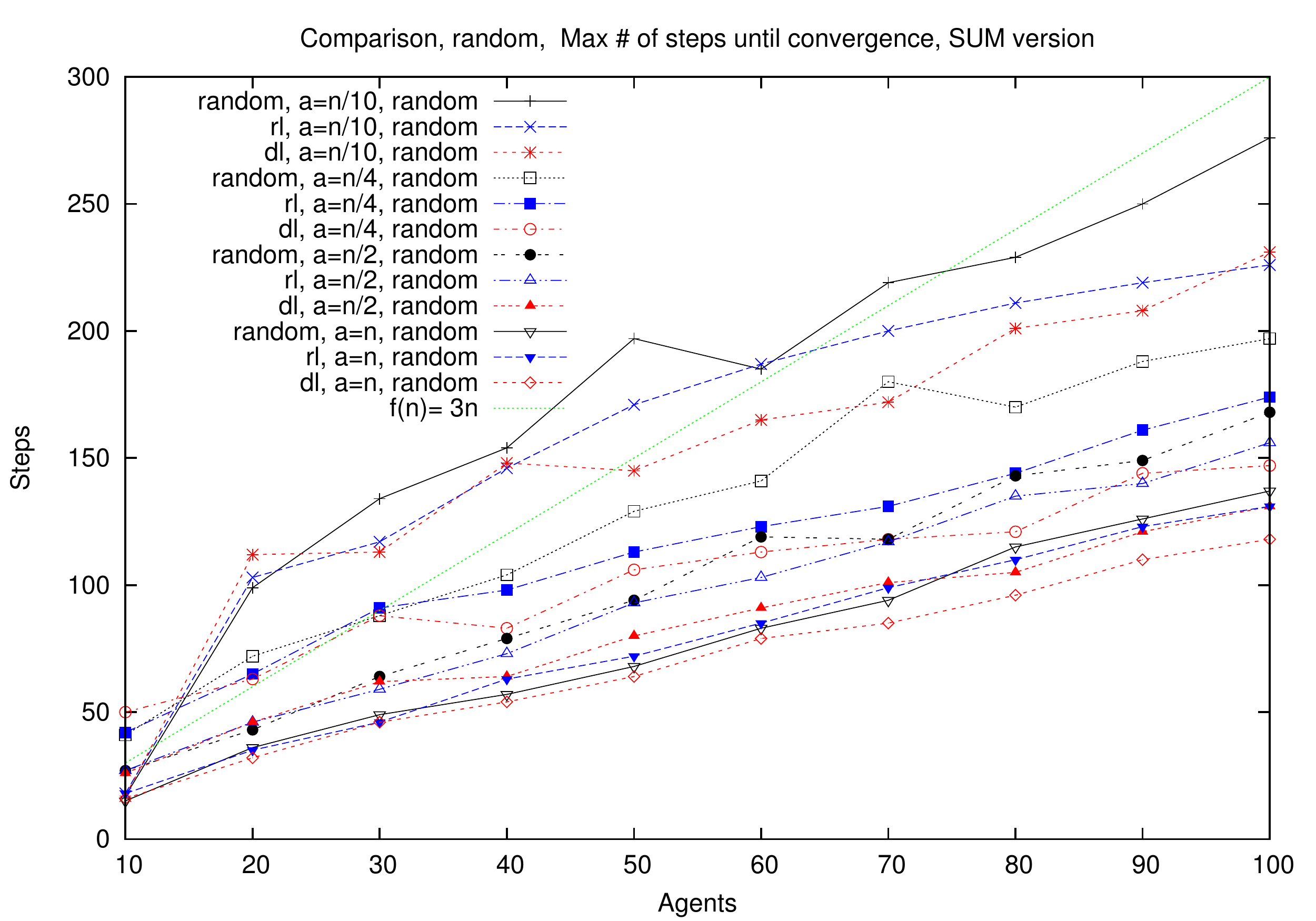}
  \caption{Comparison of different starting topologies for the \textsc{Sum}-GBG.
  The maximum number of steps needed for convergence under the max cost policy is plotted on the left, the maximum number of steps needed for convergence under the random policy is plotted on the right. Each point is the maximum over the number of steps needed for convergence of 5000 trials where the initial setting is \texttt{random}, \texttt{rl} or \texttt{dl} and $\alpha = a$.}
    \label{fig:exp_compare_sum}
\end{figure}

Generally we find that the specific topology of the initial network has only marginal impact on the number of steps needed for convergence in the \textsc{Sum}-GBG. As shown in Fig.~\ref{fig:exp_compare_sum}, the convergence time differs roughly by a factor of at most $2$. We expected that the \texttt{random} setting is faster than the \texttt{rl} setting and that \texttt{rl} is faster than \texttt{dl}. The reason for our expectation was that the initial networks in the \texttt{random} setting are more star-like than in \texttt{rl} or \texttt{dl} (which actually is the opposite of star-like) and being star-like seems to be closer to the typical shape of a stable network. However, our simulations show exactly the opposite behavior for all configurations. Under both move policies we see that \texttt{dl} is faster than \texttt{rl} and \texttt{random}. One possible explanation is that due to the high diameter of the initial network in the \texttt{dl} setting a move by some agent, especially in the first rounds, 
decreases the social cost by a larger amount than in the \texttt{rl} or \texttt{random} setting. Maybe this initial decrease is large enough to reduce the number of moves at the end of the process.  

Another interesting point is the comparison of the two move policies. We find that the max cost policy outperforms the random policy, independently of the initial setting. But there are differences:
Under the max cost policy we see that the \texttt{dl} setting is the fastest, independently of the edge-price $\alpha$. The convergence time in the other settings depends stronger on $\alpha$ and they are very similar. The reason for this may be that only a few moves by max cost agents in the \texttt{rl} suffice to obtain a network which is very similar to the initial network in the \texttt{random} setting. Under the random policy, we again see for each configuration that \texttt{dl} is faster than \texttt{rl}, which is faster than \texttt{random}. However, we also see that the convergence time depends stronger on $\alpha$ than under the max cost policy. It seems that the value of $\alpha$ has a stronger influence than the respective initial setting. This is not surprising since the max cost policy favors agents whose move decreases the social cost by a rather large amount which leads to faster convergence.

 \paragraph{Results for the \textsc{Max}-GBG:} 
 
 The results of our simulations for the \textsc{Max}-GBG are shown in Fig.~\ref{fig:exp_greedy_max}. We generally find the same behavior as for the \textsc{Sum}-version.
 \begin{figure}[!h]
 \centering
  \includegraphics[width=0.495\textwidth]{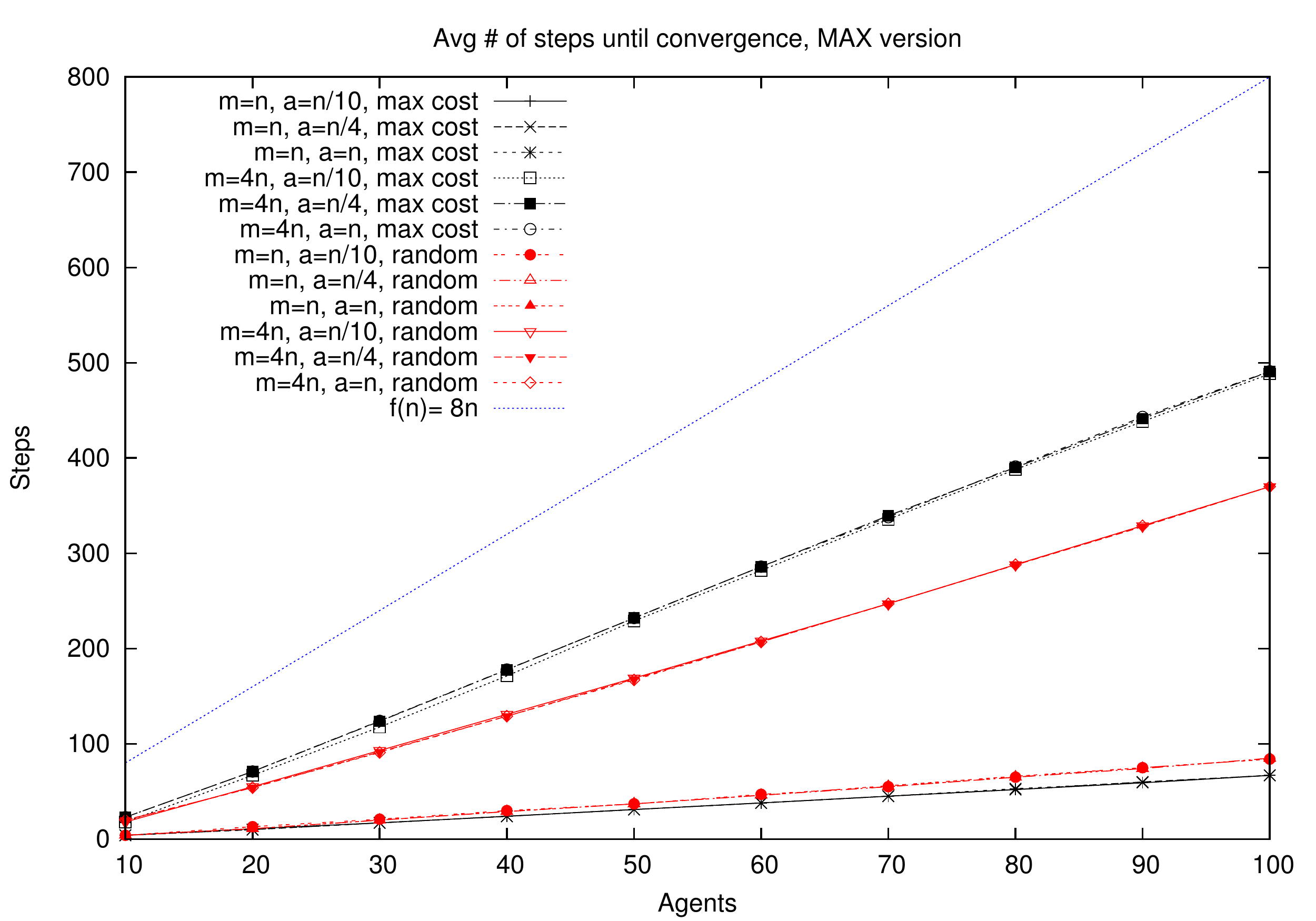}
  \includegraphics[width=0.495\textwidth]{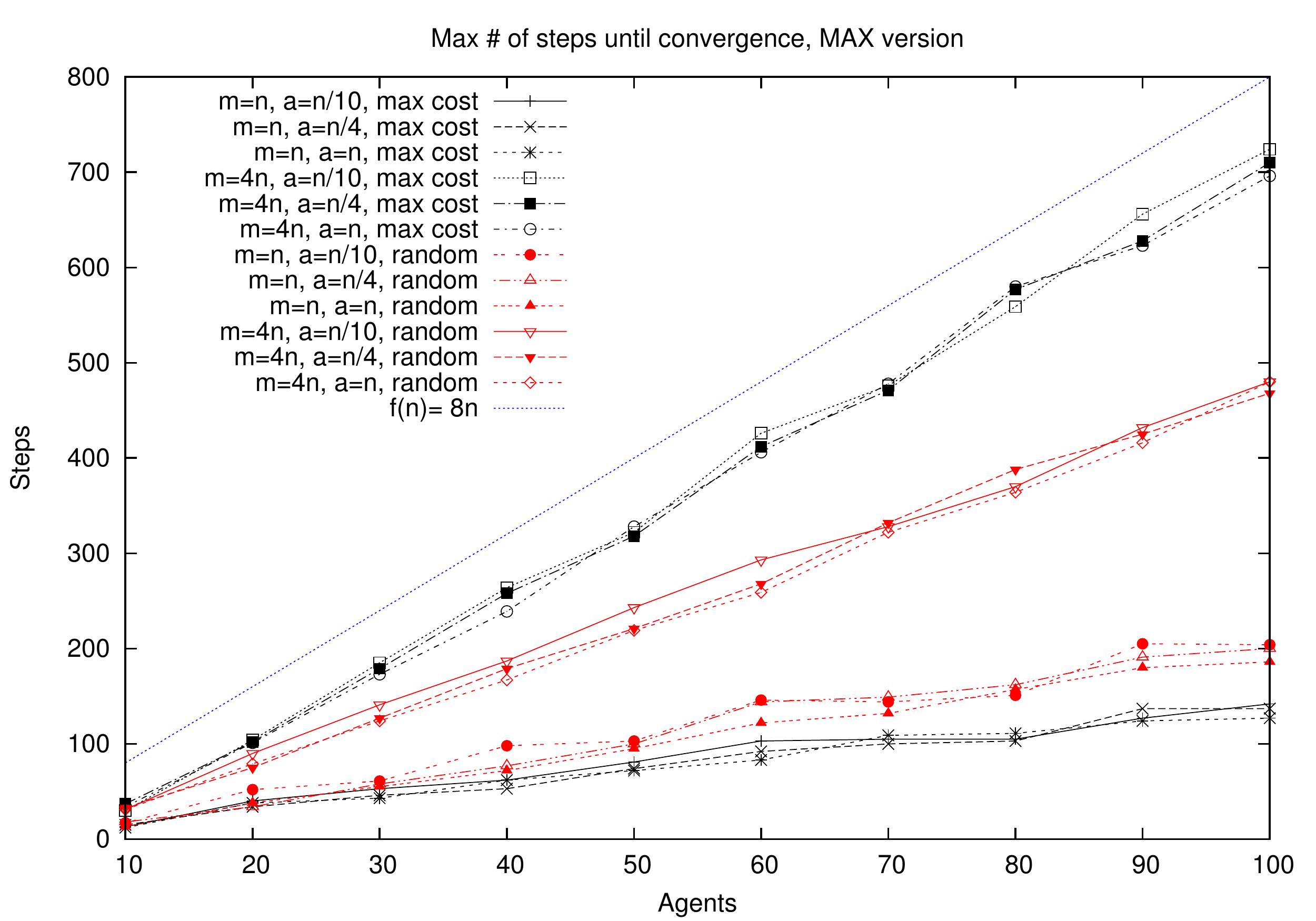}
  \caption{Experimental results for the \textsc{Max}-GBG. The average number of steps needed for convergence is plotted on the left, the maximum number of steps needed for convergence is plotted on the right. Each point is the average/maximum over the number of steps needed for convergence of 5000 trials with random initial networks having $m$ edges and $\alpha = a$.}
    \label{fig:exp_greedy_max}
\end{figure}
The convergence time grows linear in $n$ but, at least for the cases where $m=4n$ it takes longer than the respective configurations in the \textsc{Sum}-version. 

One difference is, that the choice of $\alpha$ seems to have less impact on the convergence time than in the \textsc{Sum}-version. Furthermore, for the cases where $m=4n$ and $m=2n$ (omitted in the plot) we see that the process under the max cost policy is slower than under the random policy. Thus, we see the opposite behavior under both move policies.  

A typical sample trajectory for $m=4n$, with $\alpha = n/4$ under the random policy looks much like the trajectory for the \textsc{Sum}-version, but it seems that the final phase is missing. The first phase contains almost exclusively deletions and the second phase contains mostly swaps and some deletions. Add-operations are exceptionally rare, which seems to be due to the high value of $\alpha$. Interestingly, under the max cost policy a typical trajectory looks different. Here we have that the second phase is much longer and almost all moves in this phase are swaps. This may explain why the max cost policy induces a slower convergence towards a stable network. Clearly, in the \textsc{Max}-version the cost of an agent is mostly determined by her edge-cost, which explains why first a series of deletions and then mostly swap operations can be seen.  

As for the \textsc{Sum}-version, we considered different initial topologies. Fig.~\ref{fig:exp_compare_max} shows a comparison of the three settings \texttt{random}, \texttt{rl} and \texttt{dl}, which we have introduced above.
\begin{figure}[!h]
 \centering
  \includegraphics[width=0.495\textwidth]{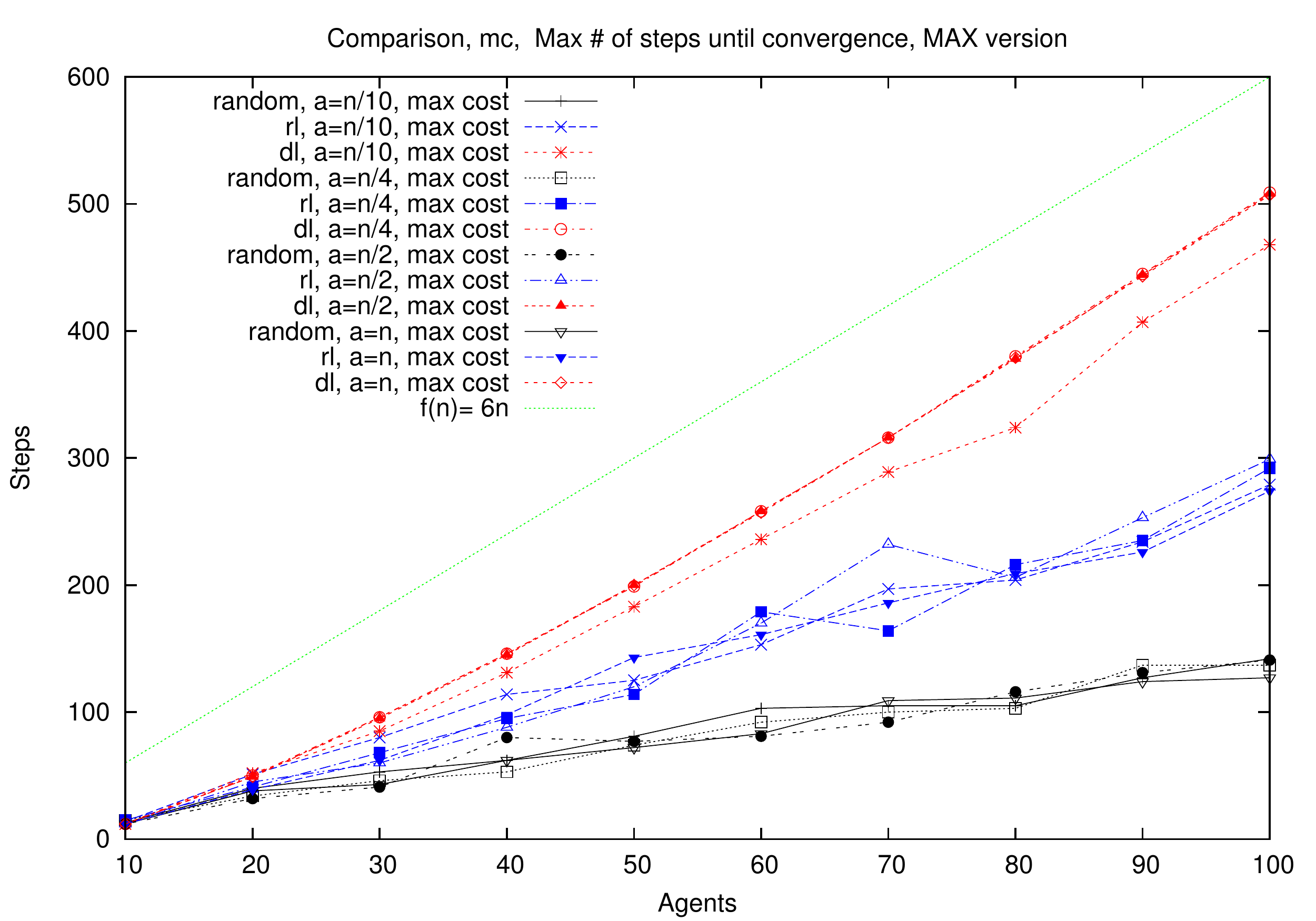}
  \includegraphics[width=0.495\textwidth]{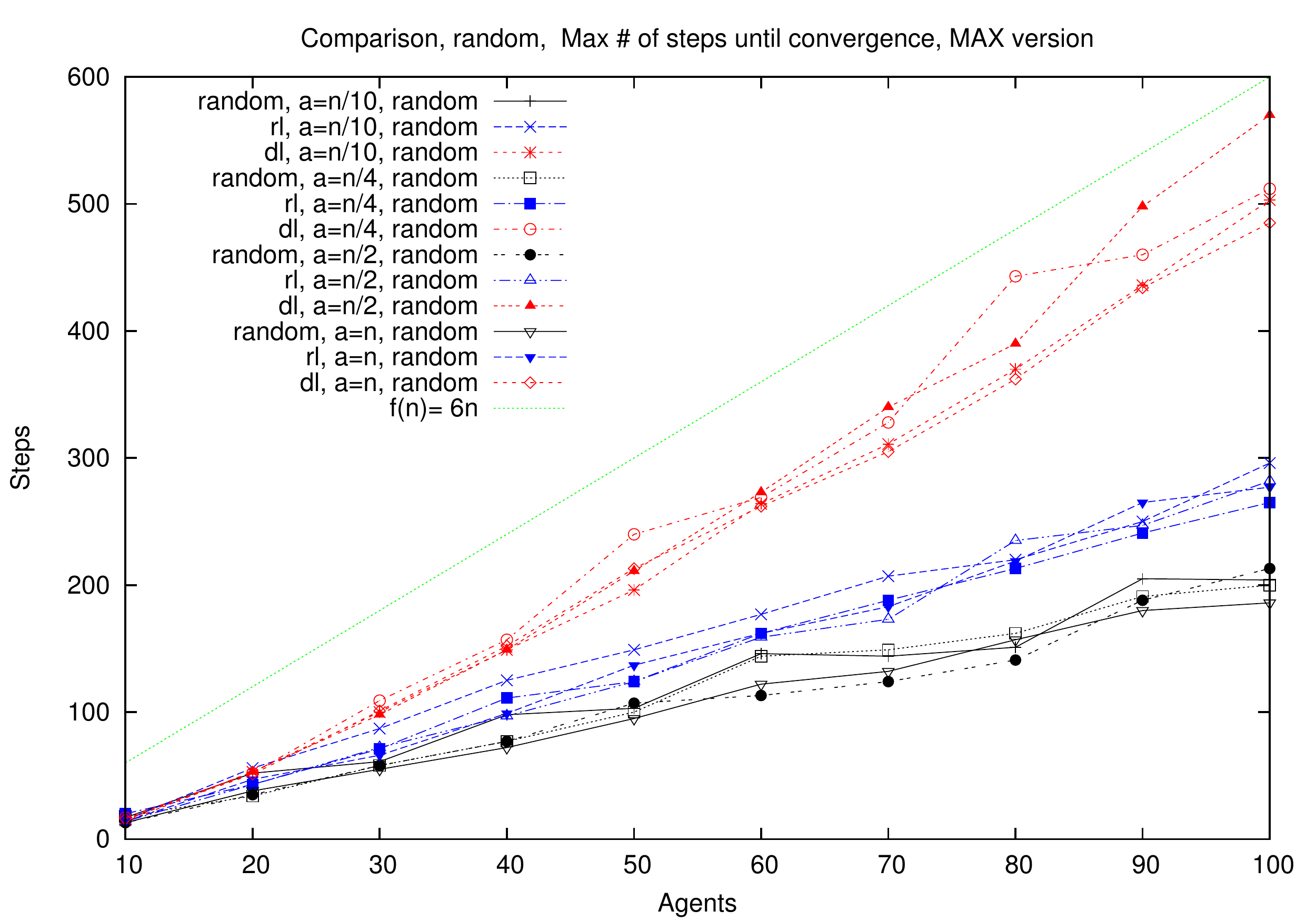}
  \caption{Comparison of different starting topologies for the \textsc{Max}-GBG.
  The maximum number of steps needed for convergence under the max cost policy is plotted on the left, the maximum number of steps needed for convergence under the random policy is plotted on the right. Each point is the maximum over the number of steps needed for convergence of 5000 trials where the initial setting is \texttt{random}, \texttt{rl} or \texttt{dl} and $\alpha = a$.}
    \label{fig:exp_compare_max}
\end{figure}

The comparison shows a stronger impact of the starting topology than in the \textsc{Sum}-version. In the \textsc{Max}-version we find that the convergence times differ by at most a factor of $5$. Interestingly, here we find the expected outcome that under both move policies the \texttt{random} setting is faster than \texttt{rl} and \texttt{dl} and that \texttt{rl} is faster than \texttt{dl}. Furthermore, the edge-price $\alpha$ has almost no influence on the convergence time compared to the influence of the initial topology. Both move polices yield almost the same convergence times. As argued above, this is not surprising, since we have that under the max cost policy a few moves suffice to obtain a network where a large fraction of agents has maximum cost and then we are almost in the random policy scenario. 

\paragraph{Further Remarks:} As for the ASG, we have not encountered a cyclic instance for the \textsc{Sum}-GBG or the \textsc{Max}-GBG in any of our several million trials. This indicates that non-convergent initial networks may be very rare. It is not impossible that the random policy or even the max cost policy may guarantee convergence. However, as our result for the GBG on general host-graphs shows, there may be initial networks on the complete host-graph for which no sequence of improving moves leads to a stable network. Finding such initial networks seems to be very challenging, but - in contrast - proving guaranteed convergence for some move policy seems even more challenging.

\section{Dynamics in Bilateral Buy Games with Cost-Sharing}\label{sec_bilateral_equal_split}
We consider ``bilateral network formation'', as introduced by Corbo and Parkes~\cite{CP05}, which we call the \emph{bilateral equal-split BG}. This version explicitly models that bilateral consent is needed in order to create an edge, which is a realistic assumption in some settings. The cost of an edge is split equally among its endpoints and edges are build only if \emph{both} incident agents are willing to pay half of the edge-price. This model implicitly assumes coordination among coalitions of size two and the corresponding solution concept is therefore the pairwise Nash equilibrium, which can be understood as the minimal coalitional refinement of the pure Nash equilibrium. The authors of~\cite{CP05} show that this solution concept is equivalent to Meyerson's proper equilibrium~\cite{M97}, which implies guaranteed convergence if the agents repeatedly play best response strategies against \emph{perturbations} of the other players' strategies, where costly mistakes are made with less probability. We show 
in 
this section that these perturbations are necessary for achieving convergence by proving that the bilateral equal-split BG is not weakly acyclic in the \textsc{Sum}-version and that it admits best response cycles in the \textsc{Max}-version. Interestingly, the first result is stronger than the result for the \textsc{Sum}-(G)BG, which yields the counter-intuitive observation, that sharing the cost of edges can lead to worse dynamic behavior.  

We want to avoid that agents are forced to buy edges. Hence, we assume that the initial network of the network creation process is connected. This assumption may be removed by enforcing a penalty for every disconnected agent, as proposed by Fabrikant et al.~\cite{Fab03} and analyzed by Brandes et al.~\cite{BHN08}. 

The creation of an edge in bilateral equal-split BGs requires coordination of the two incident agents, whereas the deletion of an edge is a unilateral move of one agent. Formally, the edge-cost term in the cost function of an agent changes slightly, since in this version an agent has to pay half of the price for \emph{every} incident edge. 
Let $G$ be any network and let $N_G(u)$ be the set of neighbors of agent $u$ in $G$. Let $G'$ be the network induced by a strategy change of agent $u$ in network $G$ and let $N_{G'}(u)$ be the set of agent $u$'s neighbors in $G'$. Thus, agent $u$ changes her strategy from $N_G(u)$ to $N_{G'}(u)$. Such a strategy change is \emph{feasible} if and only if $c_G(v) \geq c_{G'}(v)$ for all $v \in N_{G'}(u) \setminus N_G(u)$. 
Hence, agent $u$ can perform a move from strategy $S_u$ to $S_u'$ if and only if in the network induced by $S_u'$ \emph{all} agents involved in the creation of \emph{new} edges selfishly agree to pay their cost-share. If this is not the case for some agent $x \in S_u'\setminus S_u$, then we say that agent $x$ \emph{blocks} agent $u$'s move from $S_u$ to $S_u'$.

\begin{theorem}\label{thm_sum_equal_split}
 The \textsc{Sum} bilateral equal-split Buy Game is not weakly acyclic.
\end{theorem}
\begin{proof}{Theorem~\ref{thm_sum_equal_split}}
 We prove the statement by giving a cyclic sequence of networks $G_0,\dots,G_2$, where all unhappy agents in network $G_i$ only have feasible improving moves which lead to a network which is isomorphic to network $G_{i+1 \bmod 3}$. It follows that starting from network $G_0$ as initial network \emph{no} sequence of improving moves leads to a network where all agents are happy. 
 \begin{figure}[!h]
  \centering
  \includegraphics[width=13cm]{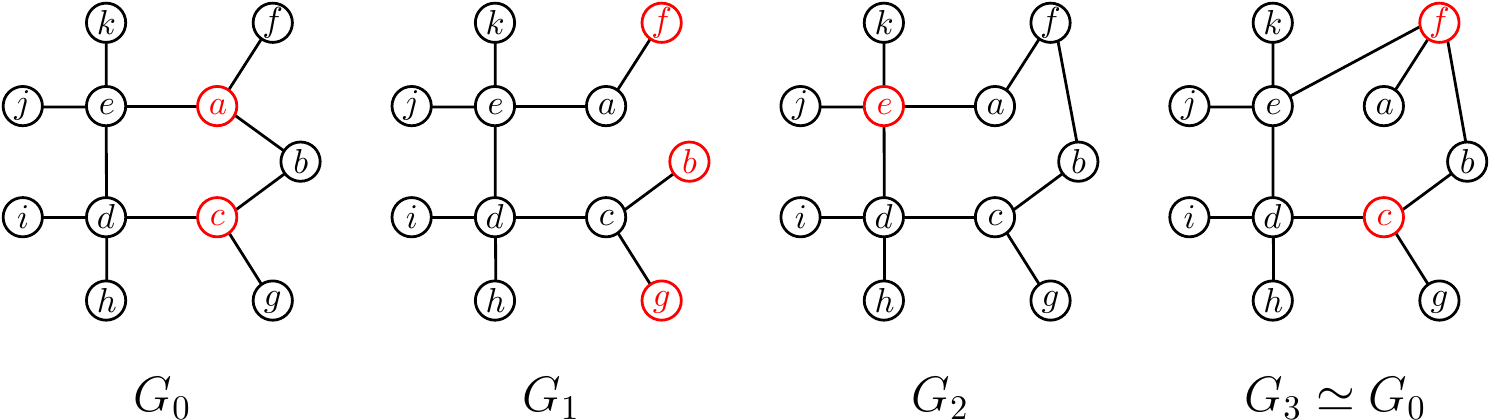}
  \caption{The cyclic sequence of networks $G_0,G_1,G_2$ for $10 < \alpha < 12$ for the \textsc{Sum} bilateral equal-split Buy Game. Unhappy agents are highlighted.}
  \label{fig:sum_equalsplit_br}
\end{figure}
 
 \noindent Now we analyze the states of the cycle. The networks $G_0,G_1$ and $G_2$ are shown in Fig.~\ref{fig:sum_equalsplit_br}. We assume that $10<\alpha<12$ holds. 
 
 \textbf{Network $G_0$:} We show that only agent $a$ and $c$ are unhappy in $G_0$ and that their unique improving move is the removal of their edge towards agent $b$, respectively. After this move, we obtain a network which is isomorphic to network $G_1$.
 
 Consider the leaf agents $f,g,h,i,j$ and $k$. Clearly, none of them can delete an edge since this would disconnect the network. Next, swapping their unique edge involves another agent who must be willing to accept another leaf neighbor, but, since $\tfrac{\alpha}{2}>1$, no agent will do this. It follows, that a leaf agent must buy at least two edges if she wants to improve. Among those edges, there cannot be edges towards other leaf agents. This is true because the network $G_0$ has diameter $4$ and no vertex has more than two neighboring leaves. At best, such an edge can decrease the distance towards the other leaf by $3$, towards the neighbor of this leaf by $1$ and possibly towards the other leaf connected to this neighbor by $1$, which yields at most a distance decrease of $5 < \tfrac{\alpha}{2}$. All non-leaf agents have distance at most $3$ towards any leaf agent and distance at most $2$ towards any non-leaf agent. It follows that no non-leaf vertex will accept any edge offered by a leaf agent, 
independently of which other edges this leaf agent buys. This implies that no leaf agent can perform any strategy change to decrease her cost. 
 
 Now we focus on agent $d$ and $e$. By symmetry, both agents face the same situation and therefore we restrict our attention to agent $d$. Her cost in $G_0$ is $4\tfrac{\alpha}{2}+17$. Clearly, agent $d$ cannot remove her edge towards $h$ or $i$ and not both edges to $c$ and $e$, since this would disconnect the network. Removing edge $dc$ or $de$ alone increases her distance-cost by $7$ or $14$, respectively. Since $\tfrac{\alpha}{2} < 7$, these moves are not improving for agent $d$. Thus, agent $d$'s best possible strategy buys at least three edges including the edges $dh$ and $di$. Among all strategies of agent $d$ which buy exactly three edges, the strategy $\{a,h,i\}$ is optimal, since $a$ has minimum distance-cost in the network $G_0 - \{d,i,h\}$. Note, that the strategy $\{a,h,i\}$ is better than agent $d$'s current strategy $\{c,e,h,i\}$ and it is the only strategy involving three edges with this property. 
 However, the move from $\{c,e,h,i\}$ to $\{a,h,i\}$ will be blocked by agent $a$, since she currently 
has cost $3\tfrac{\alpha}{2} + 20$ and with agent $d$'s new strategy her cost increases to $4\tfrac{\alpha}{2} + 17$. Agent $d$'s current strategy is best possible among all strategies which buy four edges. This is easy to see, since the edges $dh$ and $di$ are forced and the other two edges should connect to the vertices of a $2$-median-set in the graph $G_0-\{d,h,i\}$. There are two such sets: $\{c,e\}$ and $\{b,e\}$. Thus, the strategy $\{b,e,h,i\}$ yields the same cost as strategy $\{c,e,h,i\}$ for agent~$d$. We claim that no strategy buying more than four edges can outperform agent $d$'s current strategy. Any such strategy buys at least five edges and no edge will connect to a leaf of $G_0$. The strategy $\{a,c,e,h,i\}$ is the best strategy using five edges and yields cost $5\tfrac{\alpha}{2} + 15$. With six edges there is only one strategy which yields cost of $6\tfrac{\alpha}{2} + 14$. Thus, agents $d$ and $e$ cannot perform an improving move in $G_0$.
 
 Next, we show that agent $b$ is happy in network $G_0$. Clearly, agent $b$ has to buy at least one edge, since otherwise the network is disconnected. Among all these strategies, the strategies $\{d\}$ and $\{e\}$, which both yield cost $\tfrac{\alpha}{2} + 25$, are optimal. Note, that both outperform agent $b$'s current strategy $\{a,c\}$ having cost $2\tfrac{\alpha}{2} + 22$. However, in both cases, agent $b$ cannot change towards the new strategy, since the respective new neighbor will block this change. Consider $b$'s strategy change from $\{a,c\}$ to $\{d\}$. Before the change, agent $d$ has cost $4\tfrac{\alpha}{2}+ 17$, after the change $d$'s cost is $5\tfrac{\alpha}{2} + 16$, which is strictly larger. Note that no other strategy of agent $b$ which buys exactly one edge outperforms her current strategy. By deleting her edge to $a$ or $c$ her distance-cost increases by $11$. Furthermore, connecting to a leaf is clearly worse than agent $b$'s current strategy. By the same reasoning as above, agents $d$ 
and $e$ will block every new strategy of $b$ which connects to them. Since connecting to a leaf is clearly sub-optimal, we are left with $b$'s current strategy $\{a,c\}$. Thus, agent $b$ cannot perform an improving move. 
 
 Agents $a$ and $c$ are unhappy and, by symmetry, we will focus on agent $a$, who has cost $3\tfrac{\alpha}{2}+20$ in $G_0$. Clearly, all optimal strategies of agent $a$ must buy the edge towards $f$ and at least one additional edge. The best possible strategy using two edges is $\{d,f\}$, since $d$ is a $1$-median vertex of the graph $G_0 - \{a,f\}$. But, agent $e$ would block this strategy, since this strategy-change yields cost $5\tfrac{\alpha}{2} + 15 > 4\tfrac{\alpha}{2}+17$ for agent $e$. There are two other strategies using two edges and which outperform agent $a$'s current strategy. These strategies are $\{c,f\}$ and $\{e,f\}$. Moving from $\{b,e,f\}$ to $\{c,f\}$ is not possible for agent $a$, since this move would be blocked by agent $c$, whose cost would change from $3\tfrac{\alpha}{2}+20$ to $4\tfrac{\alpha}{2}+18$. The move from $\{b,e,f\}$ to $\{e,f\}$ is possible, since the move only consists of the deletion of an edge, which is a unilateral move. 
 This move decreases agent $a$'s cost from $3\tfrac{\alpha}{2}+20$ to $2\tfrac{\alpha}{2} + 25$, which is indeed an improvement since $\tfrac{\alpha}{2}>5$. Now we are left to show, that agent $a$ has no other improving moves. Clearly, since connecting towards leaf agents cannot outperform $a$'s current strategy, we can ignore all other strategies which buy two edges. For all larger strategies the same holds true. We have already shown that agent $c$ blocks any connection attempts of agent $a$. By analogous reasoning, the same is true for agents $d$ and $e$. It follows, that there are no other possible improving moves for agent $a$. Hence, we have that $a$'s move from $\{b,e,f\}$ to $\{e,f\}$ is her only improving move. By symmetry the same holds for agent $c$'s move from $\{b,d,g\}$ to $\{d,g\}$. In both cases we obtain a network which is isomorphic to network $G_1$.
 
 \textbf{Network $G_1$:} We show that in network $G_1$ only the agents $b,f,g$ are unhappy and that any improving move of one of those agents leads to a network which is isomorphic to network~$G_2$. 
 
 Analogous to the discussion above, we have that no non-leaf agent is willing to accept an edge from agents $h,i,j$ or $k$ and that any edge towards another leaf agent does not yield a distance decrease which is high enough to compensate the edge-price. Hence, these agents cannot perform a strategy-change. 
 
 Agents $d$ and $e$ are in a very similar situation than in $G_0$ and, since they cannot buy less edges, it is easy to see, that they cannot perform an improving move. Agent $a$ is happy, since she has just performed her only improving move which transformed network $G_0$ into network $G_1$. If $a$ would have another improving move, then this would contradict the uniqueness of her move in $G_0$. 
 
 Agent $c$ cannot delete her edges to $b$ and $g$ and has to buy at least three edges. Her best possible strategy using three edges would connect to $e$ instead of $d$, but $c$ cannot move towards this strategy, since agent $e$ would block it. Connecting to $a$ instead of $d$ is clearly worse. Any strategy which uses four edges cannot outperform agent $c$ current strategy. 
 Since agent $e$ refuses an edge from $c$, there are only agent $a$ or leaf agents left. Again, connecting to leaf agents yields a cost increase. Connecting to agent $a$ yields a decrease in distance-cost of $4$ but, since $\tfrac{\alpha}{2} > 4$, it follows that the strategy $\{a,b,d,g\}$ is worse than agent $c$'s current strategy. More than $4$ edges would involve an edge towards a leaf agent and therefore we conclude that agent~$c$ cannot improve on her current strategy $\{b,d,g\}$. 

We are left with agents $b,f$ and $g$. Agents $b$ and $g$ face a similar situation an we will focus on agent $b$ having cost $\tfrac{\alpha}{2} + 33$. Clearly, agent $b$ has to buy at least one edge and, among all strategies using exactly one edge, agent $b$'s current strategy $\{c\}$ is the only possible strategy, since all other vertices would block an edge-swap from agent $b$ to them. Hence, we consider possible strategies of agent $b$, which buy at least two edges. Agents $d$ and $e$ will refuse any edge from $b$, independently of the number of other edges bought by agent $b$. Thus, $d$ and $e$ cannot be involved in any feasible improving strategy of agent $b$. Next, by buying an edge towards $h,i,j$ or $k$ agent $b$ cannot decrease her distance-cost more than $4 < \tfrac{\alpha}{2}$, which rules out that $h,i,j$ or $k$ are involved in an improving strategy. We are left with vertices $a,c,f$ and $g$ as possible targets for edges from $b$ and we have that $b$'s strategy has to choose at least two of them. 
With these restrictions it is easy to see, that buying at least three edges is too expensive for agent $b$. Hence, we focus on all possible strategies using exactly two edges. Clearly, strategy $\{a,c\}$ is the best possible among them, but it requires that agents $a$ and $c$ are willing to accept $b$'s edge. For agent $c$ this is true, since in $G_1$ this edge is already present and because $bc$ is a bridge. Unfortunately for $b$, agent $a$ will block the move from $\{c\}$ to $\{a,c\}$, since $a$ was the agent who unilaterally decided to remove the edge towards $b$ in network $G_0$. The same is true for a move towards strategy $\{a,g\}$. We are left with the two strategies $\{c,f\}$ and $\{f,g\}$. After moving to strategy $\{f,g\}$ agent $b$ has cost $2\tfrac{\alpha}{2}+28$, which, since $\tfrac{\alpha}{2}>5$, is higher than agent $b$'s current cost of $\tfrac{\alpha}{2}+33$. 
Finally, the move towards strategy $\{c,f\}$ will be an improving move, since this strategy yields cost $2\tfrac{\alpha}{2}+25 < \tfrac{\alpha}{2}+33$. Furthermore, the move from $\{c\}$ to $\{c,f\}$ will not be blocked by agent $f$, since this move decreases agent $f$'s cost from $\tfrac{\alpha}{2} + 34$ to $2\tfrac{\alpha}{2} + 26$. Moreover, the move will not be blocked by agent $c$, since the move strictly decreases agent $c$'s cost. Thus, we have that the move from $\{c\}$ to $\{c,f\}$ is agent $b$'s unique feasible improving move. By symmetry, the same is true for agent $g$'s move from strategy $\{c\}$ to strategy $\{c,f\}$. 

Now we consider agent $f$. By analogous reasoning there is no feasible improving move for agent $f$ towards a strategy which uses exactly one or more than two edges. We are left with analyzing all possible strategies using exactly two edges. Similar to the situation of agent $b$, only the agents $a,b,c$ and $g$ are possible targets for $f$'s edges. The strategy $\{a,c\}$ looks the most promising for agent $f$, but it will be blocked by agent $c$, since $c$'s decrease in distance-cost caused by this move is only $4<\tfrac{\alpha}{2}$. The same is true for strategy $\{b,c\}$ or $\{g,c\}$. Moreover, strategy $\{b,g\}$ yields higher cost then $f$'s current strategy. Hence, we are left with the strategies $\{a,b\}$ and $\{a,g\}$. Both yield cost $2\tfrac{\alpha}{2} + 26 < \tfrac{\alpha}{2} + 33$, since $\tfrac{\alpha}{2} < 7$. A move from $\{a\}$ to $\{a,b\}$ or $\{a,g\}$ is feasible, since agent $b$'s (or $g$'s) cost changes from $\tfrac{\alpha}{2} + 31$ to $2\tfrac{\alpha}{2} + 25$, which is a strict cost 
decrease since $\frac{\alpha}{2}< 6$. Furthermore, it is easy to see that agent $a$'s cost strictly decreases if $f$ moves to strategy $\{a,b\}$ or $\{a,g\}$, which implies that $a$ will not block such a move. 

Observe, that all possible improving moves by agents $b,f$ or $g$ lead to a network which is isomorphic to~$G_2$.

\textbf{Network $G_2$:} We show that in network $G_2$ only agent $e$ is unhappy and that her unique feasible improving move leads to a network which is isomorphic to network $G_0$.

We consider the leaf agents $g,h,i,j$ and $k$ first. Clearly, they cannot delete their unique edge and no agent would accept an edge if a leaf agent performs an edge swap. Thus, they must buy at least two edges if they want to outperform their current strategy. Analogous to the discussion above, no edge towards another leaf agent can be part of an improving strategy. Thus, it remains to show that no strategy change towards a combination of at least two edges to non-leaf agents is a feasible improving move for $g,h,i,j$ or $k$. The agents $d$ and $e$ will not accept any new edge from a leaf agent since they have only one non-leaf agent in distance $3$ and any leaf agent is in distance at most $3$. It follows that by accepting such an edge agent $d$ or $e$ can only hope for a distance decrease of $3<\tfrac{\alpha}{2}$. The same is true for agents $a$ and $c$. Agent $a$ has one non-leaf agent in distance $3$ and one leaf-agent, which is $g$, in distance $4$. Thus, by accepting an edge coming from a leaf agent, 
agent $a$'s distance-cost can decrease by at most $4<\tfrac{\alpha}{2}$. For agent $c$ the situation is similar, but $c$ does not have a leaf agent in distance $4$, which implies that by accepting an edge from a leaf agent only a distance decrease of $3$ is possible. We are left with agents $b$ and $f$, which both have one non-leaf agent in distance $3$ and two leaf agents in distance $4$. Hence, by accepting an edge from a leaf agent they can possibly reduce their distance towards this leaf by $3$, towards its neighbor by $1$ and towards the other leaf in distance $4$ by $1$. In total the best possible distance decrease for agents $b$ or $f$ is $5 < \tfrac{\alpha}{2}$. It follows, that no leaf agent can perform a feasible improving move.

Next, we show that agent $d$ cannot improve on her current strategy $\{c,e,h,i\}$ which yields cost of $4\tfrac{\alpha}{2} + 17$. Clearly, agent $d$ must buy the edges towards $h$ and $i$ and at least one more edge to connect to the network $G_2 - \{d,h,i\}$. The best possible strategy using three edges connects to a $1$-median vertex of $G_2 - \{d,h,i\}$. There are two such vertices: $a$ and $f$. Both strategies $\{a,h,i\}$ and $\{f,h,i\}$ yield cost $3\tfrac{\alpha}{2} + 25$ which is higher than agent $d$'s current cost. Hence, no strategy using exactly three edges can outperform $d$'s current strategy. The optimal strategy using $4$ edges must connect towards the vertices which form a $2$-median-set in the graph $G_2 - \{d,h,i\}$. The $2$-median-problem in $G_2-\{d,h,i\}$ has two solutions: $\{c,e\}$, $\{b,e\}$. Thus, agent $d$'s current strategy is an optimal strategy using four edges. Now let us look at possible strategies for agent $d$ which use more than four edges. The best strategy using five edges 
is $\{c,
e,f,h,i\}$ and yields cost $5\tfrac{\alpha}{2} + 15 > 4\tfrac{\alpha}{2} + 17$. Clearly, the best strategy using six edges is worse. Thus, agent $d$ cannot perform any improving move.

Now, let us consider agent $c$ having cost $3\tfrac{\alpha}{2}+20$ with her current strategy $\{b,d,g\}$. Agent $c$ must buy the edge towards $g$ and at least one more edge to ensure connectedness of the network. Her best possible strategy using exactly two edges is $\{e,g\}$, since $g$ is the unique $1$-median-vertex of $G_2-\{c,g\}$. However, the move from $\{b,d,g\}$ towards $\{e,g\}$ will be blocked by agent $e$ since this move increases her cost from $4\tfrac{\alpha}{2}+18$ to $5\tfrac{\alpha}{2}+16$. The two second best strategies using two edges are $\{a,g\}$ and $\{d,g\}$ which both yield cost $2\tfrac{\alpha}{2}+26 > 3\tfrac{\alpha}{2}+20$ since $\tfrac{\alpha}{2}<6$. Hence, no strategy using two edges can be a feasible improving strategy for agent $c$. For all other strategies using more than two edges we have that they cannot contain an edge towards $e$. This is true since in $G_2$ agent $e$ only has two agents in distance $3$, which implies that by accepting an edge from $c$ agent $e$ can only 
hope for a distance decrease by $3<\tfrac{\alpha}{2}$. It follows that $c$'s best possible strategy using three edges must connect to agent $d$ and $g$. The best choice for the third edge is agent $f$, but the move from $\{b,d,g\}$ to $\{d,f,g\}$ will be blocked by agent $f$ since this move changes her cost from $2\tfrac{\alpha}{2} + 26$ to $3\tfrac{\alpha}{2} + 21$ which is a strict cost increase, since $5<\tfrac{\alpha}{2}$. Clearly, the third edge cannot connect to a leaf agent of $G_2$. Hence, we are left with $c$'s current strategy and the strategy $\{a,d,g\}$, which yield the same cost. It follows, that $c$'s current strategy is the best possible feasible strategy using three edges. If agent $c$ would buy more than three edges, then, since $e$ is not available and leaf agents are not attractive as well, the best such strategy connects to $d$ and $g$ and chooses two targets from the set $\{a,b,f\}$. It is easy to see that such a strategy yields higher cost than $c$'s current strategy. Moreover, if $c$ 
buys 
more than four edges, then the situation gets even worse. Thus, we have that agent $c$ cannot perform a feasible improving move.

Agent $b$ has cost $2\tfrac{\alpha}{2}+25$ in network $G_2$. Agent $b$ cannot move towards a strategy using exactly one edge, since the removal of the edge $bf$ increases $b$'s distance-cost by $6 > \tfrac{\alpha}{2}$ and removing the edge $bc$ yields an increase in distance-cost by $16$. Furthermore, no other agent than $c$ or $f$ would accept an edge from agent $b$ if this edge is $b$'s unique edge.  Thus, we have that agent $b$ has to buy at least two edges to outperform her current strategy. Note, that agents $d$ and $e$ will refuse to accept any edge offered by agent $b$, since they have at most two agents in distance $3$ an can only hope for a distance decrease by $3$ from such an edge. By buying an edge towards a leaf agent, agent $b$ can only hope for a distance decrease by $5<\tfrac{\alpha}{2}$, since $b$ has two agents in distance $4$ and their common neighbor is in distance $3$. Agent $a$ will refuse an edge from $b$, since this edge must yield a distance decrease for $a$ by at least $6$ but this 
is only possible if $b$ simultaneously buys edges towards $c,g,h$ and $i$, which clearly is not an improving strategy for agent $b$. Hence, agent $b$ only has agents $c$ and $f$ available as targets and connecting to both is $b$'s current strategy. It follows that agent $b$ has no feasible improving move. 

Agent $f$ has cost $2\tfrac{\alpha}{2} + 26$. We first show that $f$ has no feasible improving strategy using one edge. Removing edge $fa$ or $fb$ increases $f$'s distance-cost by $11$ or $6$, respectively. Hence, no such removal yields a cost decrease. Furthermore, it is easy to see that no other agent than $a$ and $b$ would accept an edge from $f$ if $f$ buys no other edges. By an analogous argument as for agent $b$, it follows that no edge towards a leaf agent is beneficial for $f$ and that no non-leaf agent other than $a$ and $b$ would accept an edge from $f$. Thus, agent $f$ cannot perform any improving strategy change. 

Next, we show that agent $a$, having cost $2\tfrac{\alpha}{2}+23$, cannot move towards an improving strategy. No move towards a strategy using one edge can be feasible and yield a cost decrease. Removing the edge $ae$ or $af$ yields an increase in distance-cost by $18$ or $6$, respectively, which implies that agent $a$ would not improve. Moreover, no agent other than $e$ and $f$ would accept an edge from $a$ if $a$ buys no other edges. For all strategies which buy more than one edge, we have that agent $d$ would refuse to accept any edge coming from $a$, since $d$ only has one agent in distance $3$ and this implies that $d$ could only gain $2$ in distance-cost. Moreover, no edge towards a leaf agent can be part of an improving strategy for $a$, since $a$ has only one leaf in distance $4$ and could decrease her distance-cost by at most $4$ by such an edge. Hence, all edges of $a$ must connect to vertices $b,c,e$ or $f$ and it is obvious, that in $a$'s best possible strategy the edge towards $e$ should by 
contained. With this restriction it follows that $a$'s best possible strategy must use exactly two edges, since the strategies $\{b,c,e\}$, $\{c,e,f\}$, $\{b,e,f\}$ and $\{b,c,e,f\}$ are clearly more expensive than $a$'s current strategy. The strategy $\{c,e\}$ outperforms $a$'s current strategy, but the move from $\{e,f\}$ to $\{c,e\}$ will be blocked by agent $c$ since this move increases her cost from $3\tfrac{\alpha}{2}+20$ to $4\tfrac{\alpha}{2}+18$. Analogously, agent $a$'s move from $\{e,f\}$ to $\{b,e\}$ will be blocked by agent $b$ since this move increases her cost from $2\tfrac{\alpha}{2}+25$ to $3\tfrac{\alpha}{2}+21$. Thus, we have that agent $a$ cannot perform a feasible improving move. 

Finally, we show that agent $e$, having cost $4\tfrac{\alpha}{2}+18$, is unhappy in $G_2$. Clearly, agent $e$ must buy the edges towards $j$ and $k$ and at least one additional edge. Since $c$ is the unique $1$-median vertex in $G_2-\{e,j,k\}$, we have that $\{c,j,k\}$ is agent $e$'s best possible strategy which buys three edges. This strategy outperforms $e$'s current strategy, but the move from $\{a,d,j,k\}$ to $\{c,j,k\}$ will be blocked by agent $c$ since this move increases her cost from $3\tfrac{\alpha}{2}+20$ to $4\tfrac{\alpha}{2}+17$. The second best strategies using exactly three edges, $\{d,j,k\}$ and $\{b,j,k\}$, both yield cost $3\tfrac{\alpha}{2}+24$ for agent $e$. Since $\tfrac{\alpha}{2}<6$, this implies that these strategies yield higher cost for agent $e$. It follows that no move to a strategy with three edges can be feasible and improving for agent $e$. The best possible strategy using four edges must connect to the vertices in a $2$-median-set of $G_2-\{e,j,k\}$. This set is $\{d,f\}$ and 
it is 
unique, hence we have that $\{d,f,j,k\}$ is agent $e$'s best possible strategy using four edges. Note that this strategy yields cost $4\tfrac{\alpha}{2}+17$ which implies that it outperforms agent $e$'s current strategy. Furthermore, the move from $\{a,d,j,k\}$ towards $\{d,f,j,k\}$ is feasible, since this move decreases agent $f$'s cost from $2\tfrac{\alpha}{2}+26$ to $3\tfrac{\alpha}{2}+20$. This is indeed a strict decrease, since $\tfrac{\alpha}{2}<6$. We have found a feasible improving strategy change for agent $e$. In the following we will show that this is the only feasible improving strategy change. The second best strategies using exactly four edges are $\{a,d,j,k\}$, which is $e$'s current strategy and $\{b,d,j,k\}$. Both yield the same cost for agent $e$. It follows, that there are no other possible improving strategies using four edges. Strategies using more than four edges cannot outperform agent $e$'s current strategy. This can be seen as follows. With $e$'s current strategy we have that there 
are 
two agents in distance $3$ and both do not share a common neighbor. The best possible situation with more than four edges would be to have five vertices in distance $1$ and the rest in distance $2$. This yields a cost of $5\tfrac{\alpha}{2}+15 > 4\tfrac{\alpha}{2}+18$. Thus, we have shown that agent $e$ can perform exactly one feasible improving strategy change and this move transforms network $G_2$ into a network (labeled $G_3$ in Fig.~\ref{fig:sum_equalsplit_br}) which is isomorphic to $G_0$.
\end{proof}

For the \textsc{Max}-version, we can show a slightly weaker result.
\begin{theorem}\label{thm_max_equal_split}
 The \textsc{Max} bilateral equal-split Buy Game admits best response cycles.
\end{theorem}
\begin{proof}{Theorem~\ref{thm_max_equal_split}}
The four steps $G_1,\dots, G_4$ of the best response cycle are depicted in Fig.~\ref{fig:max_equalsplit_br}. We assume that $2 < \alpha < 4$ holds.
 \begin{figure}[!h]
  \centering
  \includegraphics[width=12cm]{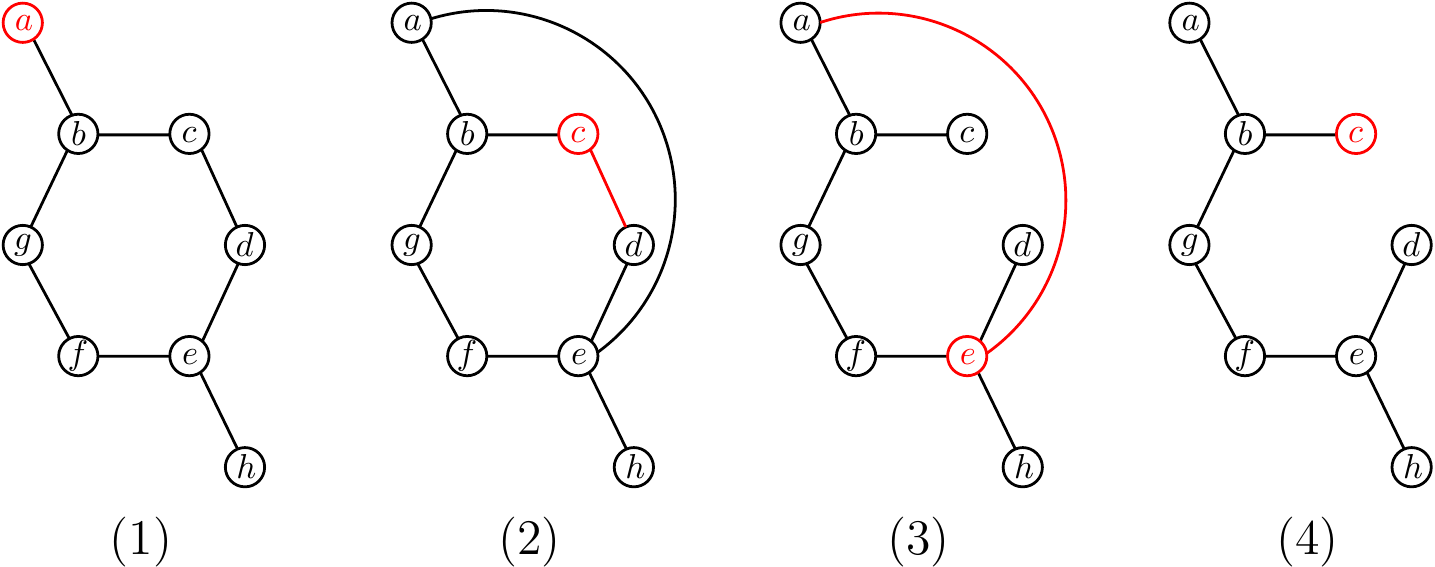}
  \caption{The steps of a best response cycle for $2 < \alpha < 4$ for the \textsc{Max} bilateral equal-split Buy Game.}
  \label{fig:max_equalsplit_br}
\end{figure}

\noindent In network $G_1$, shown in Fig.~\ref{fig:max_equalsplit_br}~(1), we claim that agent $a$, having cost $\tfrac{\alpha}{2}+5$ is unhappy with her situation and that buying the edge $ae$ is the best possible feasible strategy change for her. With strategy $\{b,e\}$ agent $a$ has cost $2\tfrac{\alpha}{2}+2 < \tfrac{\alpha}{2}+5$, since $\tfrac{\alpha}{2}<2$. Furthermore, agent $a$'s move from $\{b\}$ to $\{b,e\}$ strictly decreases agent $e$'s cost from $3\tfrac{\alpha}{2}+4$ to $4\tfrac{\alpha}{2}+2$.

Observe, that agent $a$ cannot remove any edge. By swapping her unique edge, agent $a$ can possibly achieve distance-cost of $4$, but we have $2\tfrac{\alpha}{2}+2 < \tfrac{\alpha}{2}+4$. Thus, no strategy using one edge can outperform her move from $\{b\}$ to $\{b,e\}$. Note, that by buying less than seven edges, the best possible distance-cost agent $a$ can hope for is $2$. Since the strategy $\{b,e\}$ already achieves this, it follows that agent $a$'s move from $\{b\}$ to $\{b,e\}$ is indeed a best possible strategy change. This move transforms network $G_1$ into network $G_2$.

In network $G_2$, shown in Fig.~\ref{fig:max_equalsplit_br}~(2), we claim that agent $c$'s best possible feasible strategy change is the removal of edge $cd$. By removing this edge, agent $c$'s cost changes from $2\tfrac{\alpha}{2}+3$ to $\tfrac{\alpha}{2}+4$, which is a strict decrease since $\tfrac{\alpha}{2}>1$. Among all strategies which buy one edge, the strategy $\{e\}$ would be optimal for agent $c$ since this is the only strategy which yields distance-cost of $3$. However, a move from $\{b,d\}$ to $\{e\}$ will be blocked by agent $e$, whose cost changes from $4\tfrac{\alpha}{2}+2$ to $5\tfrac{\alpha}{2}+2$. It follows that distance-cost of $4$ is best possible if agent $c$ buys only one edge. Among all strategies using two edges, the strategy $\{b,e\}$ is the only one which yields distance-cost $2$. But, as we have already seen, agent $e$ would block agent $c$'s move from $\{b,d\}$ to $\{b,e\}$. It follows that $c$'s current strategy is the best possible among all strategies which buy two edges. By 
buying more than two and less than seven edges, agent $c$ can only hope for distance-cost $2$, but even with three edges, this yields higher cost than her current strategy. Buying seven edges is clearly too expensive. Hence, we have that the removal of edge $cd$ is a feasible and best possible strategy change for agent $c$ and this change transforms $G_2$ into $G_3$.

Agent $e$ is unhappy in network $G_3$, shown in Fig.~\ref{fig:max_equalsplit_br}~(3). We show that $e$'s best possible feasible strategy change is the removal of edge $ea$. This move decreases agent $e$'s cost from $4\tfrac{\alpha}{2}+3$ to $3\tfrac{\alpha}{2}+4$, which is a strict decrease since $1<\tfrac{\alpha}{2}$. In any improving strategy agent $e$ must buy an edge towards $d$ and $h$ and at least one additional edge. Clearly, agent $e$'s best possible strategies using three edges is $\{b,d,h\}$ and $\{d,g,h\}$, since $b$ and $g$ is are the $1$-center vertices of $G_3-\{d,e,h\}$. Both strategies outperform agent $e$'s removal of $ea$, but a move from $\{a,d,f,h\}$ to $\{b,d,h\}$ will be blocked by agent $b$ whose cost will be increased from $3\tfrac{\alpha}{2}+3$ to $4\tfrac{\alpha}{2}+2$ and a move from $\{a,d,f,h\}$ to $\{d,g,h\}$ will be blocked by agent $g$, since her cost increases from $2\tfrac{\alpha}{2}+3$ to $3\tfrac{\alpha}{2}+2$. Furthermore, strategies $\{b,d,h\}$ and $\{d,g,h\}$ are the 
only strategies using three edges, which yield a distance-cost of $3$ for agent $e$. The same reasoning applies for strategies which buy more than three edges. In this case agents $b$ and $g$ would refuse to accept any edge from agent $e$, since they can only hope for a distance-cost of $2$, which is not low enough to compensate for the additional edge-cost. It follows, that agent $e$'s move from $\{a,d,f,h\}$ to $\{d,f,h\}$ is her best possible feasible strategy change. This move transforms network $G_3$ into network $G_4$.

In network $G_4$, shown in Fig.~\ref{fig:max_equalsplit_br}~(4), we claim that agent $c$ is unhappy and that her best possible feasible strategy change is the move from strategy $\{b\}$ to $\{b,d\}$, which decrease her cost from $\tfrac{\alpha}{2}+5$ to $2\tfrac{\alpha}{2}+3$, which is indeed a strict decrease since $\tfrac{\alpha}{2}<2$. Any strategy change towards a strategy using only one edge, that is, any edge swap by agent $c$, will be blocked by the other involved agent. This is true, since no agent has agent $c$ as her unique agent in maximum distance. It follows, that $c$ has to move to a strategy which buys at least two edges, if she wants to outperform her current strategy. The best possible strategy with two edges is $\{b,e\}$, since this set is the unique $2$-median-set in the graph $G_4-\{c\}$ which yields a distance-cost of $2$ for agent $c$. Unfortunately for agent $c$ the move from $\{b\}$ to $\{b,e\}$ will be blocked by agent $e$, since this move would increase her cost from $3\tfrac{\alpha}
{2}+4$ to $4\tfrac{\alpha}{2}+3$. No other strategy using two edges yields distance-cost of $2$ for agent $c$, which implies, that she has to buy more than $2$ edges to outperform her move to $\{b,d\}$. It is easy to see, that with at least three edges and less than seven edges at best a distance-cost of $2$ is possible for agent $c$, but this does not suffice to compensate the higher edge-cost. Furthermore, buying seven edges is clearly too expensive. Thus, the indicated move from $\{b\}$ to $\{b,d\}$ is a best possible feasible move for agent $c$. This move transforms $G_4$ into $G_1$ and we have completed the best response cycle.
\end{proof}




\bibliographystyle{plain}
\bibliography{dynamics}


\end{document}